\documentclass{acmsmall} 
\usepackage{amsmath,amsfonts,amssymb}
\usepackage{bm}
\usepackage{epsfig}
\usepackage{centernot}
\usepackage{bm}
\usepackage{xspace}
\usepackage{paralist}
\usepackage{graphicx}
\usepackage{subfig}
\usepackage{boxedminipage}
\usepackage[pagebackref,colorlinks=true,pdfpagemode=none,urlcolor=blue,linkcolor=blue,citecolor=blue,pdfstartview=FitH]{hyperref}
\newtheorem{claim}[theorem]{Claim}

\newtheorem{observation}[theorem]{Observation}

\def\R{\RR}
\def\RR{\mathbb R}
\def\ZZ{\mathbb Z}
\def\eqdef{~\triangleq~}
\def\eps{\varepsilon}

        {\hspace*{\fill}$\Box$\par}
\def\B{{\mathbf D}}
\def\Hyp{{\sf Hyp}}
\def\G{{\mathsf G}}
\def\d{{\sf d}}
\def\a{{\mathfrak a}}
\def\b{{\mathfrak b}}
\def\c{{\mathfrak c}}

\newcommand{\hcd}{\mathsf{hcd}}
\newcommand{\VG}{VG}
\newcommand{\md}[1]{\ (\operatorname{mod} #1)}
\newcommand{\cupdot}{\sqcup}

\newcommand{\cB}{\mathcal{B}}

\newcommand{\cL}{{\mathcal{L}}}
\newcommand{\cM}{\mathcal{M}}
\newcommand{\cP}{\mathcal{P}}

\newcommand{\divides}{\mid}
\newcommand{\notdivides}{\centernot\mid}
\newcommand{\nd}{\notdivides}
\def\M{\mathbf M}
\def\H{\mathbf H}
\def\D{\mathbf D}
\newcommand{\bS}{{\bf S}}
\newcommand{\bll}{{\bf 1}}
\newcommand{\Sec}[1]{\hyperref[sec:#1]{\S\ref*{sec:#1}}} 
\newcommand{\Eqn}[1]{\hyperref[eq:#1]{(\ref*{eq:#1})}} 
\newcommand{\Fig}[1]{\hyperref[fig:#1]{Fig.\,\ref*{fig:#1}}} 
\newcommand{\Tab}[1]{\hyperref[tab:#1]{Tab.\,\ref*{tab:#1}}} 
\newcommand{\Thm}[1]{\hyperref[thm:#1]{Theorem\,\ref*{thm:#1}}} 
\newcommand{\Lem}[1]{\hyperref[lem:#1]{Lemma\,\ref*{lem:#1}}} 
\newcommand{\Prop}[1]{\hyperref[prop:#1]{Prop.~\ref*{prop:#1}}} 
\newcommand{\Cor}[1]{\hyperref[cor:#1]{Corollary~\ref*{cor:#1}}} 
\newcommand{\Def}[1]{\hyperref[def:#1]{Definition~\ref*{def:#1}}} 
\newcommand{\Alg}[1]{\hyperref[alg:#1]{Alg.~\ref*{alg:#1}}} 
\newcommand{\Ex}[1]{\hyperref[ex:#1]{Ex.~\ref*{ex:#1}}} 
\newcommand{\Clm}[1]{\hyperref[clm:#1]{Claim~\ref*{clm:#1}}} 

\def\D{\mathbf D}
\def\R{\mathbf R}

\newcommand{\hd}[2]{{||#1 - #2||_1}}
\newcommand{\proj}{{\mathsf{proj}}}

\newcommand{\ignore}[1]{}

\newcommand{\st}[2]{st_{#2}(#1)}
\newcommand{\cross}[2]{cr_{#2}(#1)}
\newcommand{\sk}[2]{sk_{#2}(#1)}

\def\doctype{2}
\def\tsubmission{1}

\ifnum\doctype=\tsubmission
	\newcommand{\full}[1]{}
	\newcommand{\submit}[1]{#1}
\else
	\newcommand{\full}[1]{#1}
	\newcommand{\submit}[1]{}
\fi
\begin{document}
\title{Optimal bounds for monotonicity and Lipschitz testing over 
hypercubes and hypergrids}
\author{D. CHAKRABARTY
\affil{Microsoft Research} 
C. SESHADHRI
\affil{Sandia National Laboratories, Livermore} 
}
\begin{abstract}
The problem of monotonicity testing over the hypergrid and its special case, the hypercube, is a classic, well-studied, yet unsolved question in property testing. We are given query access to $f:[k]^n \mapsto \R$
(for some ordered range $\R$). 
The hypergrid/cube has a natural partial order given by coordinate-wise ordering, denoted by $\prec$.
A function is \emph{monotone} if for all pairs $x \prec y$,
$f(x) \leq f(y)$. The distance to monotonicity, $\eps_f$, is the minimum fraction of values of $f$
that need to be changed to make $f$ monotone.  
For $k=2$ (the boolean hypercube), the usual tester is the \emph{edge tester}, which
checks monotonicity on adjacent pairs of domain points. 
It is known that the edge tester using $O(\eps^{-1}n\log|\R|)$
samples can distinguish a monotone function from one where $\eps_f > \eps$.
On the other hand, the best lower bound for monotonicity testing over general $\R$ is $\Omega(n)$. 
We resolve this long standing open problem and prove that $O(n/\eps)$ samples suffice
for the edge tester. For hypergrids, existing testers require $O(\eps^{-1}n\log k\log |\R|)$ samples. 
We give
a (non-adaptive) monotonicity tester for hypergrids running in $O(\eps^{-1} n\log k)$ time,
recently shown to be optimal.
Our techniques lead to optimal property testers (with the same running time) for the natural \emph{Lipschitz property}
on hypercubes and hypergrids. (A $c$-Lipschitz function
is one where $|f(x) - f(y)| \leq c\|x-y\|_1$.)
In fact, we give a general unified proof for $O(\eps^{-1}n\log k)$-query testers for a class of ``bounded-derivative" properties that contains both monotonicity and Lipschitz.
\end{abstract}
\begin{bottomstuff}
A preliminary version of this result appeared as \cite{ChSe13}.
\end{bottomstuff}
\maketitle
\section{Introduction}\label{sec:intro}
Monotonicity testing over hypergrids~\cite{GGLRS00} is a classic problem in
property testing. We focus on functions $f:\D \mapsto \R$, where the domain, $\D$, is the hypergrid $[k]^n$ and the range, $\R$, is a total order.
The hypergrid/hypercube defines the natural coordinate-wise partial order: $x \preceq y$, iff $\forall i \in [n], x_i \leq y_i$. A function $f$ is  {\em monotone} if $f(x) \leq f(y)$ whenever $x\preceq y$.
The \emph{distance to monotonicity}, denoted by $\eps_f$, 
is  the minimum fraction of places at which $f$ must be changed to have the property $\cP$. Formally,
if $\cM$ is the set of all monotone functions,
$\eps_f \eqdef \min_{g \in \cM}\left(|\{x | f(x) \neq g(x)\}|/|\D|\right).$
Given a parameter $\eps \in (0,1)$, the aim is to
design a randomized algorithm for the following problem. If $\eps_f = 0$ (meaning $f$ is monotone), 
the algorithm must accept with probability $>2/3$, and if $\eps_f > \eps$, it must reject with
probability $>2/3$. If $\eps_f \in (0,\eps)$, then any answer
is allowed. 
Such an algorithm is called a \emph{monotonicity tester}. 
The quality of a tester is determined by the number of queries to $f$.
A \emph{one-sided tester} accepts with probability $1$ if the function is monotone.
A \emph{non-adaptive tester} decides all of its queries in advance, so
the queries are independent of the answers it receives.
Monotonicity testing has been studied extensively in the past decade \cite{EKK+00,GGLRS00,DGLRRS99,LR01,FLNRRS02,AC04,E04,HK04,PRR04,ACCL04,BRW05,BGJRW09,BCG+10,BBM11}. 
Of special interest is the hypercube domain, $\{0,1\}^n$.
~\cite{GGLRS00} introduced the \emph{edge tester}. Let $\H$ be the pairs 
that differ in precisely one coordinate (the edges of the hypercube).
The edge tester picks a pair in $\H$ uniformly at random and checks if monotonicity is satisfied by this pair. For boolean range, ~\cite{GGLRS00} prove $O(n/\eps)$ samples suffice 
to give a bonafide montonicity tester. \cite{DGLRRS99} subsequently showed that $O(\eps^{-1}n\log|\R|)$ samples suffice for a general range $\R$. In the worst case, $|\R| = 2^n$, and so this gives a $O(n^2/\eps)$-query tester. The best known general lower bound is 
$\Omega(\min(n,|\R|^2))$~\cite{BBM11}.
It has been an outstanding open problem in property testing (see Question 5 in the
Open Problems list from the Bertinoro Workshop~\cite{Bert}) to give 
an optimal bound for monotonicity testing over the hypercube.
We resolve this by showing that the edge tester is indeed
optimal (when $|\R| \geq \sqrt{n}$).
\begin{theorem} \label{thm:mono-hc} 
The edge tester is a
$O(n/\eps)$-query non-adaptive, one-sided monotonicity tester for functions \submit{\\}$f:\{0,1\}^n\mapsto \R$.
\end{theorem}
For general hypergrids $[k]^n$,~\cite{DGLRRS99} give a $O(\eps^{-1}n\log k\log |\R|)$-query monotonicity tester.
Since $|\R|$ can be as large as $k^n$, this gives a $O(\eps^{-1}n^2\log^2k)$-query tester. 
In this paper, we give a $O(\eps^{-1}n\log k)$-query monotonicity tester on hypergrids that generalizes the edge tester. 
This tester is also a uniform pair tester, in the sense it defines a set $\H$ of pairs, picks a pair uniformly at random from it, and checks for monotonicity among this pair. The pairs in $\H$ also differ in exactly one coordinate, as in the edge tester. 
\begin{theorem} \label{thm:mono-hg}
There exists a non-adaptive, one-sided \submit{\\}$O(\eps^{-1}n\log k)$-query monotonicity tester for functions\submit{\\} $f:[k]^n\mapsto \R$.
\end{theorem}
\begin{remark}
Subsequent to the conference version of this work, the authors proved 
a $\Omega(\eps^{-1}n\log k)$-query lower bound for monotonicity testing on the hypergrid 
for any (adaptive, two-sided error) tester~\cite{ChSe13-2}. Thus, both the above theorems
are optimal.
\end{remark}
A property that has been studied recently is that of a function being {\em Lipschitz}: 
a function $f:[k]^n \mapsto \R$ is called $c$-Lipschitz if for all $x,y\in[k]^n, |f(x) - f(y)| \leq c\|x-y\|_1$.
The Lipschitz testing question was introduced by~\cite{JR11}, who show that for the range $\R = \delta \ZZ$,
$O((\delta\eps)^{-1}n^2)$ queries suffice for Lipschitz testing. 
For general hypergrids,~\cite{AJMS12} recently give an $O((\delta\eps)^{-1}n^2k\log k)$-query tester
for the same range.
\cite{BlJh+13} prove a  lower bound of $\Omega(n\log k)$ queries for 
non-adaptive monotonicity testers (for sufficiently large $\R$).
We give a tester for the Lipschitz property that improves all known results and matches
existing lower bounds. Observe that the following holds for arbitrary ranges.
\begin{theorem} \label{thm:lip-hc} 
There exists a non-adaptive, one-sided \submit{\\}$O(\eps^{-1}n\log k)$-query $c$-Lipschitz tester for functions\submit{\\} $f:[k]^n\mapsto \R$.
\end{theorem}
Our techniques apply to a class of properties that contains monotonicity and Lipschitz. We call it the bounded derivative property, or more technically, the $(\alpha,\beta)$-Lipschitz property.
Given parameters $\alpha,\beta$, with $\alpha < \beta$, we say that a function $f:[k]^n \mapsto \R$ 
has the $(\alpha,\beta)$-Lipschitz property if for any $x\in [k]^n$, and $y$ obtained by increasing exactly 
one coordinate of $x$ by exactly $1$, we have $\alpha \leq f(y)-f(x) \leq \beta$. Note that when $(\alpha=0,\beta=\infty)$\footnote{If the reader is uncomfortable with the choice of $\beta$ as $\infty$, $\beta$ can be thought
of as much larger than any value in $f$.}, 
we get monotonicity. When $(\alpha=-c,\beta=+c)$, we get $c$-Lipschitz. 
\begin{theorem}\label{thm:main}
There exists a non-adaptive, one-sided\submit{\\} $O(\eps^{-1}n\log k)$-query $(\alpha,\beta)$-Lipschitz tester for functions\submit{\\} $f:[k]^n\mapsto \R$, for any $\alpha<\beta$.
There is no dependence in the running time on $\alpha$ and $\beta$.
\end{theorem}
Although \Thm{main} implies all the other theorems stated above, we prove \Thm{mono-hc} and \Thm{mono-hg} before giving a whole proof of \Thm{main}. 
The final proof is a little heavy on notation, and the proof of the monotonicity theorems illustrates the new techniques.
\subsection{Previous work} \label{sec:prev}
We discuss some other previous work on monotonicity testers for hypergrids.
For the 
total order (the case $n=1$), which has been called the monotonicity testing problem on the {\em line}, \cite{EKK+00} give a $O(\eps^{-1}\log k)$-query tester, and this is optimal \cite{EKK+00,E04}. 
Results for general posets were first obtained by~\cite{FLNRRS02}.
The elegant concept of $2$-TC spanners introduced by \cite{BGJRW09} give a general class of monotonicity testers for various posets.
It is known that such constructions give testers with polynomial dependence of $n$ for the hypergrid \cite{BGJ+12}.
For constant $n$,~\cite{HK04,AC04} give a $O(\eps^{-1}\log k)$-query tester (although the dependency on $n$ is exponential). 
From the lower bound side,~\cite{FLNRRS02} first prove an $\Omega(\sqrt{n})$ 
(non-adaptive, one-sided) lower bound for hypercubes. 
\cite{BCG+10} give an $\Omega(n/\eps)$-lower bound for non-adaptive, one-sided testers,
and a breakthrough result of~\cite{BBM11} prove a general $\Omega(\min(n,|\R|^2)$ lower bound.
Testing the Lipschitz property is a natural question that arises in many applications. For instance, given a computer program, one may like to test the robustness of the program's output to the input. This has been studied before, for instance in \cite{CGLN11}, however, the solution provided looks into the code to detect if the program satisfies Lipschitz or not. The property testing setting is a black-box approach to the problem. \cite{JR11} also provide an application to differential privacy; a class of mechanisms known as Laplace mechanisms proposed by \cite{DMNS06} achieve privacy in the process of outputting a function by adding a noise proportional to the Lipschitz constant of the function. 
\cite{JR11} gave numerous results on Lipschitz testing over hypergrids.
They give a 
$O(\eps^{-1}\log k)$-query tester for the 
line, a general $\Omega(n)$-query lower bound for the Lipschitz testing question on the hypercube, and 
a non-adaptive, 1-sided $\Omega(\log k)$-query  lower bound on the line. 
\section{The Proof Roadmap} \label{sec:min}
The challenge of property testing is to relate the tester behavior
to the distance of the function to the property. Consider monotonicity over the hypercube. 
To argue about the edge tester, we want to show that
a large distance to monotonicity implies many violated edges. Most current
analyses of the edge tester go via what we could call the \emph{contrapositive route}.
If there are few violated edges in $f$, then they show the distance to monotonicity
is small. This is done by modifying $f$ to make it monotone, and bounding the number of changes as a function of the 
number of violated edges.
There is an inherently ``constructive" viewpoint to this: it specifies a method
to convert non-monotone functions to monotone ones.
Implementing this becomes difficult when the range is large, and existing bounds degrade
with $\R$. For the Lipschitz property, this route becomes incredibly complex.
A non-constructive approach may give more power, but how
does one get a handle on the distance? The \emph{violation graph} provides a method.
The violation graph has $[k]^n$ as the vertex set and an edge between
any pair of comparable domain vertices $(x,y)$ ($x \prec y$) if $f(x) > f(y)$.
The following theorem can be found as Corollary 2 in \cite{FLNRRS02}.
\begin{theorem}[\cite{FLNRRS02}]\label{thm:vc}
The size of the minimum vertex cover of the violation graph is exactly $\eps_f|\D|$.
As a corollary, the size of any maximal matching in the violation graph is at least $\frac{1}{2}\eps_f|\D|$.
\end{theorem}
Can a large matching in the violated graph imply there are many violated edges? 
\cite{LR01} give an approach by reducing the monotonicity testing problem on the hypercube 
to routing problems.
For any $k$ source-sink pairs on the {\em directed} hypercube, 
suppose $k\mu(k)$ edges need to be deleted in order to pairwise separate them.
Then $O(n/\eps\mu(n))$ queries suffice for the edge tester. Therefore, if $\mu(n)$ is at least a constant, one gets a linear query monotonicity tester on the cube. Lehman and Ron \cite{LR01} explicitly ask for bounds on $\mu(n)$.~\cite{BCG+10} show that $\mu(n)$ could be as small as $1/\sqrt{n}$, thereby putting an $\Omega(n^{3/2}/\eps)$ bottleneck to the above approach. 
In the reduction above, the function values are altogether ignored. More precisely, once one moves to the combinatorial routing question on source-sink pairs, the fact that they are related by actual function values is lost. 
Our analysis crucially uses the value of the functions to argue about the structure of the maximal matching in the violation graph. 
\def\M{\mathbf M}
\def\H{\mathbf H}
\def\D{\mathbf D}
\subsection{It's all about matchings}
The key insight is to move to  a \emph{weighted} violation graph. The weight of violation $(x,y)$ depends on the property at hand; for now it suffices to know that for monotonicity, the weight of $(x,y)$ ($x \prec y$) is $f(x) - f(y)$.
This can be thought of as a measure of the magnitude of the violation. 
(Violation weights were also used for Lipschitz testers~\cite{JR11}.)
We now
look at a maximum {\em weighted} matching $\M$ in the violation graph. Naturally, this
is maximal as well, so $|\M|\geq \frac{1}{2}\eps_f |\D|$. 
All our algorithms pick a pair uniformly at random 
from a predefined set $\H$ of pairs, and check the property on that pair. 
For the hypercube domain, $\H$ is the set
of all edges of the hypercube.
Our  
analysis is based on the construction of a one-to-one
mapping  from pairs in $\M$ to {\em violating} pairs in $\H$.
This mapping implies the number of violated pairs in $\H$ is at least $|\M|$, and thus the uniform pair tester succeeds with probability $\Omega(\eps_f|\D|/|\H|)$, implying  $O(|\H|/\eps_f|\D|)$ queries suffice to test monotonicity. 
For the hypercube, $|\H| = n2^{n-1}$ and $|\D| = 2^n$, giving the final bound of $O(n/\eps_f)$.
To obtain this mapping, we first decompose $\M$ into sets $M_1,M_2,\ldots,M_t$ such that each pair in $\M$ is in at least one $M_i$. Furthermore, we {\em partition} $\H$ into perfect
matchings $H_1,H_2,\ldots, H_t$.
In the hypercube case, $M_i$ is the collection of pairs in $\M$ whose $i$th coordinates differ, and $H_i$ is the collection of hypercube edges differing {\em only} in the $i$th coordinate; for the hypergrid case, the partitions are more involved.
We map each pair in $M_i$ to a unique violating pair in $H_i$. For simplicity, let us 
ignore subscripts and call the matchings $M$ and $H$. We will assume in this discussion that $M \cap H = \emptyset$. Consider the {\em alternating paths and cycles} generated by the symmetric difference of 
$\M \setminus M$ and $H$. Take a point $x$ involved in a pair of $M$, and note that it can
only be present as the endpoint of an alternating path, denoted by $\bS_x$. 
Our main technical lemma shows that each such $\bS_x$ contains a violated $H$-pair.
\begin{figure}[htbp]
\begin{center}
\includegraphics[scale=0.2]{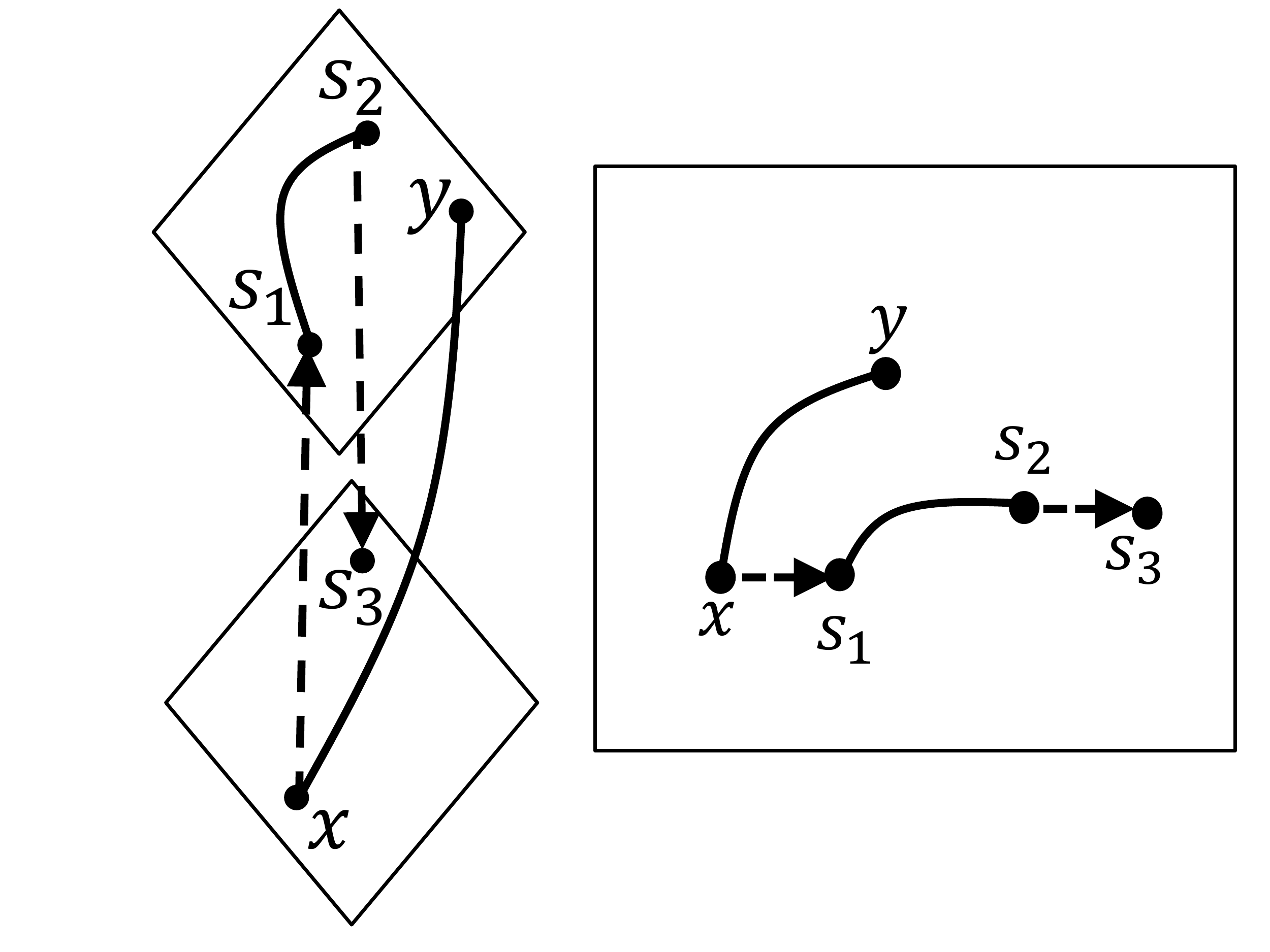}
\caption{\small{The alternating path: the dotted lines connect pairs of $M$, the solid curved lines connect pairs of $\M \setminus M$, and the dashed lines are $H$-pairs.}}
\label{fig:two-step}
\end{center}
\end{figure}
\submit{\vspace{-15pt}}
\subsection{Getting the violating $H$-pairs} \label{sec:pair}
Consider $M$, the pairs of $\M$ which differ on the $i$th coordinate, and $H$ is the 
set of edges in the dimension cut along
this coordinate. Let $(x,y)\in M$, and say $x[i] = 0$ giving us $x\prec y$.  (We denote the $a$th coordinate
of $x$ by $x[a]$.) Recall that the weight of this violation is $f(x) - f(y)$.
It is convenient to think of $\bS_x$ as follows. We begin from $x$ and take the incident
$H$-edge to reach $s_1$ (note that that $s_1 \prec y$). Then we take the $(\M\setminus M)$-pair
containing $s_1$ to get $s_2$. But what if no such pair existed? This can be possible in two ways: either $s_1$ was $\M$-unmatched or $s_1$ is $M$-matched.
If $s_1$ is $\M$-unmatched, then delete $(x,y)$ from $\M$ and add $(s_1,y)$ to obtain a new matching. If $(x,s_1)$ was not a violation, and therefore $f(x) < f(s_1)$\footnote{We are assuming here that all function values are distinct; as we show in \Clm{pert} this is without loss of generality.}, we get $f(s_1) - f(y) > f(x) - f(y)$. Thus the new matching has strictly larger weight, contradicting the choice of $\M$. If $s_1$ was $M$-matched, then let $(s_1,s_2)\in M$. First observe that 
$s_1 \succ s_2$. This is because $s_1[i] = 1$ (since $s_1[i] \neq x[i]$) and since $(s_1,s_2)\in M$ they must differ on the $i$th coordinate implying $s_2[i]=0$. 
This implies $s_2 \prec y$, and so we could
replace pairs $(x,y)$ and $(s_2,s_1)$ in $\M$ with $(s_2,y)$. 
Again, if $(x,s_1)$ is not a violation, then $f(s_2) - f(y) > [f(s_2) - f(s_1)] + [f(x) - f(y)]$, contradicting the maximality of $\M$.
Therefore, we can taje a $(\M\setminus M)$-pair to reach $s_2$. 
With care, this argument can be carried over till we find a violation, and a detailed description of this is given in \Sec{mono-hc}. 
Let us demonstrate a little further (refer to the left of \Fig{two-step}). Start with $(x,y)\in M$, and $x[i] = 0$. Following the sequence $\bS_x$,
the first term $s_1$ is $x$ projected ``up" dimension cut $H$. The second term is obtained by following
the $\M \setminus M$-pair incident to $s_1$ to get $s_2$. 
Now we claim that $s_2 \succ s_1$, for otherwise one can remove $(x,y)$ and $(s_1, s_2)$
and add $(x,s_1)$ and $(s_2, y)$ to increase the matching weight. (We just made the argument earlier; the interested reader
may wish to verify.)
In the next step, $s_2$ is projected ``down" along $H$ to get $s_3$. 
By the nature of the dimension cut $H$, $x \prec s_3$ and $s_1 \prec y$. So, if $s_3$ is unmatched
and $(s_2, s_3)$ is not a violation, we can again rearrange the matching to improve the weight. We alternately
go ``up" and ``down" $H$ in traversing $\bS_x$, because of which we can modify the pairs in $\M$ and
get other matchings in the violation graph. The maximality of $\M$ imposes additional structure,
which leads to violating edges in $H$.
In general, the spirit of all our  arguments is as follows. Take an endpoint of $M$ and start walking
along the sequence given by the alternating paths generated by $\M \setminus M$ and $H$. Naturally,
this sequence must terminate somewhere. If we never
encounter a violating pair of $H$ during the entire sequence, then we can rewire the matching
$\M$ and increase the weight. Contradiction!
Observe the crucial nature of alternating up and down movements along $H$. This happens because the first coordinate
of the points in $\bS_x$ switches between the two values of $0$ and $1$ (for $k=2$).
Such a reasoning does not hold water in the hypergrid domain. 
The structure of $\H$ needs to be more complex, and is not
as simple as a partition of the edges of the hypergrid. Consider the extreme case of the line $[k]$.
Let $2^r$ be less than $k$. We break $[k]$ into contiguous pieces of length $2^r$. We can now match
the first part to the second, the third to the fourth, etc. In other words, the pairs look like $(1,2^r+1)$, $(2, 2^r+2)$,
$\ldots$, $(2^r, 2^{r+1})$, then $(2^{r+1}+1,2^{r+1}+2^r+1)$, $(2^{r+1}+2,2^{r+1}+2^r+2)$, etc.
We can construct such matchings for all powers of $2$ less than $k$, and these will be our $H_i$'s.
Those familiar with existing proofs for monotonicity on $[k]$ will not be surprised by this set of matchings. 
All methods need to cover all ``scales" from $1$ to $k$ (achieved
by making them all powers of $2$ up to $k$). 
It can also be easily generalized to $[k]^n$.
What about the choice of $\M$? Simply choosing $\M$ to be a maximum weight matching 
and setting up the sequences $\bS_x$ does not seem to work. It suffices to look at $[k]^2$
and the matching $H$ along the first coordinate where $r=0$, so the pairs are $\{(x,x') | x[1] = 2i-1, x'[1] = 2i, x[2] = x'[2]\}$. 
A good candidate for the corresponding $M$ is the set of pairs in $\M$ that connect lower endpoints of $H$
to higher endpoints of $H$. Let us now follow $\bS_x$ as before. 
Refer to the right part of \Fig{two-step}. Take $(x,y) \in M$
and let $x \prec y$. We get $s_1$ by following the $H$-edge on $x$, so $s_1 \succ x$.
We follow the $\M \setminus M$-pair incident to $s_1$ (suppose it exists) to get $s_2$. 
It could be that $s_2 \succ s_1$. It is in $s_3$ that we see a change from the hypercube.
We could get $s_3 \succ s_2$, because there is no guarantee that $s_2$ is at the higher
end of an $H$-pair. This could not happen
in the hypercube. We could have a situation where $s_3$ is unmatched, we have not
encountered a violation in $H$, and yet we cannot rearrange $\M$ to increase the weight. For a concrete
example, consider the points as given in \Fig{two-step} with function values
$f(x) = f(s_1) = f(s_3) = 1$, $f(y) = f(s_2) = 0$.
Some thought leads to the conclusion that $s_3$ must be less than $s_2$ for any such rearrangement argument to work.
The road out of this impasse is suggested by the two observations. First, the difference in $1$-coordinates
between $s_1$ and $s_2$ must be odd. Next, we could rearrange and match $(x,s_2)$ and $(s_1,y)$ instead.
The weight may not increase, but this matching might be more amenable to the alternating path approach.
We could start from a maximum weight matching that also maximizes the number of pairs where coordinate
differences are even. Indeed, the insight for hypergrids is the definition of a \emph{potential} $\Phi$ for $\M$.
The potential $\Phi$ is obtained by summing for every pair $(x,y) \in \M$ and every coordinate $a$, 
the largest power of $2$ dividing the 
difference $|x[a] - y[a]|$. We can show that a maximum weight matching that also maximizes $\Phi$
does not end up in the bad situation above. With some addition arguments, we can generalize 
the hypercube proof. We describe this in \Sec{mono-hg}. 
\subsection{Attacking the generalized Lipschitz property} 
One of the challenges in dealing
with the Lipschitz property is the lack of direction. The Lipschitz
property, defined as $\forall x,y, |f(x) - f(y)| \leq \|x-y\|_1$, is an 
undirected property, as opposed to monotonicity. In monotonicity, a point $x$ only
``interacts" with the subcube above and below $x$, while in Lipschitz, constraints
are defined between all pairs of points. Previous results for Lipschitz testing 
require very technical and clever machinery to deal with this issue, since
arguments analogous to monotonicity do not work.
The alternating paths argument given
above for monotonicity also exploits this directionality, as can be seen by
heavy use of inequalities in the informal calculations. Observe that in the monotonicity
example for hypergrids in \Fig{two-step}, the fact that $s_3 \succ s_2$ (as opposed to $s_3 \prec s_2$)
required the potential $\Phi$ (and a whole new proof). 
A subtle point is that while the property of Lipschitz is undirected, violations
to Lipschitz are ``directed". If $|f(x) - f(y)| > \|x-y\|_1$,
then either $f(x) - f(y) > \|x-y\|_1$ or $f(y) - f(x) > \|x-y\|_1$, but never
both. This can be interpreted as a direction for violations.  In the alternating
paths for monotonicity (especially for the hypercube), the partial order relation
between successive terms follow a fixed pattern. This is crucial
for performing the matching rewiring. 
As might be guessed, the weight of a violation $(x,y)$ becomes $\max(f(x) - f(y) - \|x-y\|_1, f(y) - f(x) - \|x-y\|_1)$.
For the generalized Lipschitz problem, this is defined in terms of a pseudo-distance over
the domain. We look at the maximum weight matching as before (and use the same potential function $\Phi$).
The notion of ``direction" takes the place of the partial order relation in monotonicity.
The main technical arguments show that these directions follow a fixed pattern in the corresponding alternating paths.
Once we have this pattern, we can perform the matching rewiring argument
for the generalized Lipschitz problem.
\section{The Alternating Paths Framework}\label{sec:altpaths}
The framework of this section is applicable for all $(\alpha,\beta)$-Lipschitz properties over hypergrids.
We begin with two objects: $\M$, the matching of violating pairs, and $H$, a matching of $\D$.
The pairs in $H$ will be aligned along a fixed dimension (denote it by $r$) with the same $\ell_1$ distance, called the $H$-distance.
That is, each pair $(x,y)$ in $H$ will differ only in one coordinate and the difference will be the same for all pairs.
We now give some definitions. 
\begin{itemize}
	\item $L(H), U(H)$: Each pair $(x,y) \in H$ has a ``lower" end $x$ and an ``upper" end $y$ depending on the value of the coordinate at which they differ.
	We use $L(H)$ (resp. $U(H)$) to denote the set of lower (resp. upper) endpoints. Note that $L(H) \cap U(H) = \emptyset$.
	\item $H$-straight pairs, $\st{\M}{H}$: All pairs $(x,y) \in \M$ with both ends in $L(H)$ or both in $U(H)$.
	\item $H$-cross pairs, $\cross{\M}{H}$: All pairs $(x,y) \in \M \setminus H$ such that $x \in L(H)$, $y \in U(H)$, and the $H$-distance divides $|y[r] - x[r]|$.
	\item $H$-skew pairs, $\sk{\M}{H}= \M \setminus (\st{\M}{H} \cup \cross{\M}{H})$.
	\item $X$: A set of lower endpoints in $\cross{\M}{H} \setminus H$.
\end{itemize}
Consider the domain $\{0,1\}^n$. We set $H$
to be (say) the first dimension cut. $\st{\M}{H}$ is the set of pairs in $(x,y) \in \M$ where
$x[1] = y[1]$. All other pairs $(x,y) \in \M$ ($x \prec y$) are in $\cross{\M}{H}$ since $x[1] = 0$ and $y[1] = 1$.
There are no $H$-skew pairs. The set $X$ will be chosen differently for the applications.
We require the following technical definition of \emph{adequate matchings}.
This arises because we will use matchings that are not necessarily perfect.
A perfect matching $H$ is always adequate. 
\begin{definition} \label{def:adeq} A matching $H$ is \emph{adequate} if
for every violation $(x,y)$, both $x$ and $y$ participate in the matching $H$.
\end{definition}
\noindent
We will henceforth assume that $H$ is adequate.
The symmetric difference of $\st{\M}{H}$ and $H$ is a collection of alternating paths and cycles. 
Because $H$ is adequate and $\st{\M}{H} \cap \cross{\M}{H} = \emptyset$,
any point in $x \in X$ is the endpoint of some alternating path (denoted by $\bS_x$). 
Throughout the paper, $i$ denotes an even index, $j$ denotes an odd index, and $k$
is an arbitrary index.
\smallskip
\begin{asparaenum}
	\item The first term $\bS_x(0)$ is $x$.
	\item For even $i$, $\bS_x(i+1) = H(\bS_x(i))$.
	\item For odd $j$: if $\bS_x(j)$ is $\st{\M}{H}$-matched, $\bS_x(j+1) = \M(\bS_x(j))$. Otherwise,	terminate.
\end{asparaenum}
\medskip
We start with a simple property of these alternating paths.
\begin{proposition}\label{prop:sub}
For $k\equiv 0,3 \md 4$, $s_k \in L(H)$. For non-negative $k \equiv 1,2 \md 4$, $s_k \in U(H)$.
\end{proposition}
\begin{proof}
If $k$ is even, then $(s_{k},s_{k+1})\in H$. Therefore, either $s_k \in L(H)$ and $s_{k+1} \in U(H)$
or vice versa. If $k$ is odd, $(s_k, s_{k+1})$ is a straight pair. So $s_k$ and $s_{k+1}$ lie in the same sets. 
Starting with $s_0 \in L(H)$,
a trivial induction completes the proof.
\end{proof}
The following is a direct corollary of \Prop{sub}.
\begin{corollary} \label{cor:sub}
If $i \equiv 0 \md 4$, $s_{i} \prec s_{i+1}$. If $i \equiv 2 \md 4$, $s_{i+1} \prec s_{i}$.
\end{corollary}
We will prove that every $\bS_x$ contains a violated $H$-pair.
Henceforth, our focus is entirely on some fixed sequence $\bS_x$.
\subsection{The sets $E_-(i)$ and $E_+(i)$} 
Our proofs are based on matching rearrangements, and this
motivates the definitions in this subsection.
For convenience, 
we denote $\bS_x$ by $x = s_0, s_1, s_2, \ldots$.
We also set $s_{-1} = y$.
Consider the sequence $s_{-1}, s_0, s_1, \ldots, s_i$, for
even $i > 1$. We define
$$ E_-(i) = (s_{-1}, s_0), (s_1, s_2), (s_3, s_4), \ldots, (s_{i-1}, s_i) = \{(s_j, s_{j+1}): \textrm{$j$ odd}, -1 \leq j < i\} $$
This is simply the set of $\M$-pairs in $\bS_x$ up to $s_i$. We now define $E_+(i)$.
Think of this as follows. We first pair up $(s_{-1},s_1)$. Then, we go in order of $\bS_x$ to pair
up the rest. We pick the first unmatched $s_k$ and pair it to the first term of \emph{opposite} parity.
We follow this till $s_{i+1}$ is paired. These sets are illustrated in \Fig{lemma4point3}.
\begin{eqnarray*}
 E_+(i) & = & (s_{-1}, s_1), (s_0, s_3), (s_2, s_5), \ldots, (s_{i-4},s_{i-1}), (s_{i-2}, s_{i+1}) \\
 & = & \{(s_{-1}, s_1)\} \cup \{(s_{i'}, s_{i'+3}): \textrm{$i'$ even}, 0 \leq i' \leq i-2\} 
\end{eqnarray*}
\begin{figure}[h]
\includegraphics[trim=4cm 6cm 3cm 6cm,clip=true,scale=0.55]{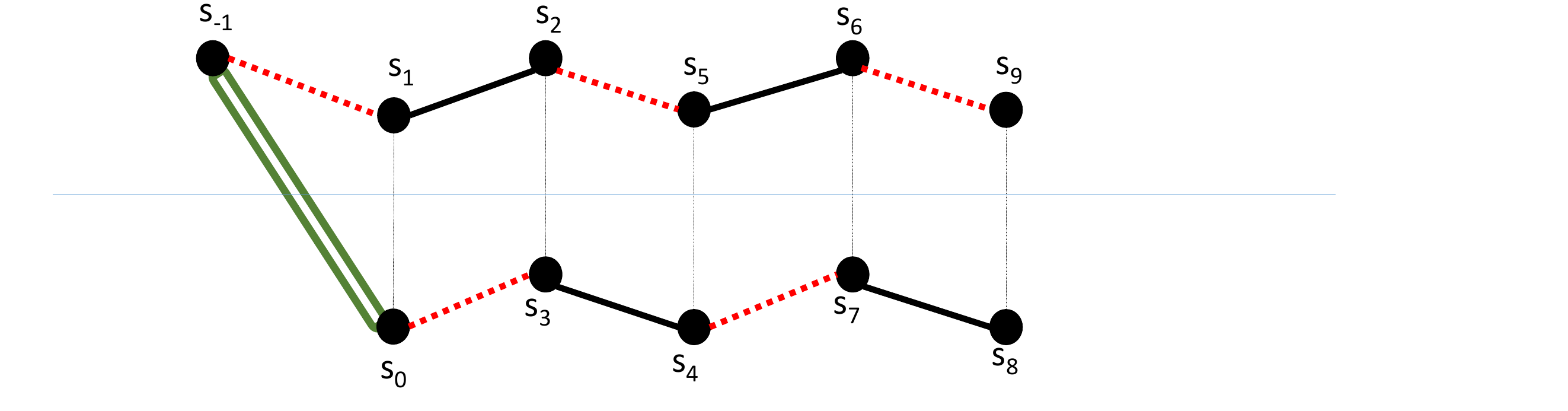}
\caption{Illustration for $i=8$. The light vertical edges are $H$-edges. The dark black ones are $\st{\M}{H}$-pairs. The green, double-lined one on the left is the starting $M$-pair. The dotted red pairs form $E_+(8)$. 
All points alove the horizonatal line are in $U(H)$, the ones below are in $L(H)$.}
\label{fig:lemma4point3}
\end{figure}
\vspace{-20pt}
\begin{proposition} \label{prop:ek} 
$E_-(i)$ involves $s_{-1}, s_0, \ldots, s_i$, 
while $E_+(i)$ involves $s_{-1}, s_0, \ldots, s_{i-1}, s_{i+1}$.
\end{proposition}
\def\bSx{\bS_x}
\section{The Structure of $\protect\bSx$ for Monotonicity} \label{sec:struct}
We now focus on monotonicity, and show that $\bS_x$ is highly structured. (The proof for general Lipschitz
will also follow the same setup, but requires more definitions.) 
The weight of a pair $(x,y)$ is defined to be $f(x)-f(y)$ if $x\prec y$, and is $-\infty$ otherwise. 
We will assume that all function values are distinct. This is without loss of generality although we prove it formally later in \Clm{pert}.
Thus violating pairs have positive weight. 
We choose a maximum weight matching $\M$ of pairs. Note that every pair in $\M$ is a violating pair.
We remind the reader
that for even $k$, $(s_k,s_{k+1}) \in H$ and for odd $k$, $(s_k,s_{k+1})\in \st{\M}{H}$. 
\subsection{Preliminary observations} 
\begin{proposition} \label{prop:cross} For all $x,y \in L(H)$ (or $U(H)$),
$x \prec y$ iff $H(x) \prec H(y)$.
Consider pair $(x,y) \in \cross{\M}{H}$ such that $x \prec y$.
Then $H(x) \prec y$ and $x \prec H(y)$.
\end{proposition}
\begin{proof} For any point in $x \in L(H)$, $H(x)$ is obtained by adding the $H$-distance
to a specific coordinate. This proves the first part.
The $H$-distance divides $|[y[r] - x[r]|$ (where $H$ is aligned in dimension $r$) and $(x,y), x \prec y$ is a cross pair.
Hence $y[r] - x[r]$ is at least the $H$-distance. Note that $H(x)$
is obtained by simply adding this distance to the $r$ coordinate of $x$, so $H(x) \prec y$.
\end{proof}
\begin{proposition} \label{prop:ek-comp} All pairs in $E_-(i)$ and $E_+(i)$ are comparable.
Furthermore, $s_{1} \prec s_{-1}$ and for all even $0 \leq k \leq i-2$, $s_{k} \prec s_{k+3}$ iff $s_{k+1} \prec s_{k+2}$.
\end{proposition}
\begin{proof}
All pairs in $E_-(k)$ are in $\M$, and hence comparable. Consider pair $(s_{-1},s_1) \in E_+(k)$.
Since $s_1 = H(s_0)$ and $(s_0,s_1)$ is a cross-pair, by \Prop{cross}, $s_1 \prec s_{-1}$.
Consider pair $(s_k, s_{k+3})$,
where $k$ is even. (Refer to \Fig{lemma4point3}.)
The pair $(H(s_k), H(s_{k+3})) = (s_{k+1}, s_{k+2})$ is in $\st{\M}{H}$. Hence, the points
are comparable and both lie in $L(H)$ or $U(H)$. By \Prop{cross},
$s_k, s_{k+3}$ inherit their comparability from $s_{k+1}, s_{k+2}$.
\end{proof}
For some even $i$, suppose $(s_{i},s_{i+1})$ is a not a violation.
\Cor{sub} implies
\begin{align} 
\textrm{If} \ i\equiv 0\md 4, f(s_{i+1}) - f(s_{i}) > 0.\nonumber\\
\textrm{If} \ i\equiv 2\md 4, f(s_{i}) - f(s_{i+1}) > 0. \label{eq:*} \tag{$*$}
\end{align}
We will also state an \emph{ordering} condition on the sequence. 
\begin{align} 
\textrm{If} \ i\equiv 0\md 4, \ s_i \prec s_{i-1}. \nonumber\\
\textrm{If} \ i\equiv 2\md 4, \ s_i \succ s_{i-1}. \label{eq:**} \tag{$**$}
\end{align}
Remember these conditions and \Cor{sub} together as follows. If $i \equiv 0 \md 4$,
$s_i$ is on smaller side, otherwise it is on the larger side.
In other words, if $i \equiv 0 \md4$,
$s_i$ is smaller than its ``neighbors" in $\bS_x$. For $i \equiv 2 \md 4$, it is bigger.
For condition \eqref{eq:*},  if $i \equiv 0 \md 4$, $f(s_i) < f(s_{i-1})$.
\subsection{The structure lemmas}
We will prove a series of lemmas that prove structural properties of $\bS_x$ that are intimately
connected to conditions \eqref{eq:*} and \eqref{eq:**}. These proofs are where much of the insight lies.
\begin{lemma}\label{lem:cons-mon}
Consider some even index $i$ such that $s_i$ exists.
Suppose conditions \eqref{eq:*} and \eqref{eq:**} held for all even indices $\leq i$.
Then, $s_{i+1}$ is $\M$-matched.
\end{lemma}
\begin{proof} The proof is by contradiction, so assume that $\M(s_{i+1})$ does not exist. 
Assume $i\equiv 0\md 4$. (The proof for the case $i\equiv 2\md 4$ is similar and omitted.)
Consider sets $E_-(i)$ and $E_+(i)$. Note that $s_{-1}, s_0, s_1, \ldots, s_{i+1}$ are all distinct.
By \Prop{ek}, $\M' = \M - E_-(i) + E_+(i)$ is a valid matching. We will
argue that $w(\M') > w(\M)$, a contradiction.
By condition \eqref{eq:**},
\begin{eqnarray}
w(E_-(i)) & = & [f(s_{0}) - f(s_{-1})] + [f(s_1) - f(s_2)] + [f(s_4) - f(s_3)] + \cdots \nonumber \\
&& \cdots + [f(s_{i-3}) - f(s_{i-2})] + [f(s_{i}) - f(s_{i-1})] \label{eq:we-}
\end{eqnarray}
By the second part of \Prop{ek-comp} (for even $k$, $s_{k} \prec s_{k+3}$ iff $s_{k+1} \prec s_{k+2}$) and condition \eqref{eq:**}, we know the comparisons for all pairs in $E_+(i)$. 
\begin{eqnarray}
w(E_+(i+2)) & = & [f(s_{1}) - f(s_{-1})] + [f(s_0) - f(s_3)] + [f(s_5) - f(s_2)] + \cdots \nonumber \\
& & \cdots +  [f(s_{i-4}) - f(s_{i-1})] + [f(s_{i+1}) - f(s_{i-2})] \label{eq:we+}
\end{eqnarray}
Note that the coefficients of common terms in $w(E_+(i))$ and $w(E_-(i))$ are identical.
The only terms not involves (by \Prop{ek}) are $f(s_{i+1})$ in $w(E_+(i))$ and $f(s_i)$ in $w(E_-(i))$.
The weight of the new matching is precisely $w(\M) - W_- + W_+ = w(\M) + f(s_{i+1}) - f(s_{i})$. 
By \eqref{eq:*} for $i$, this is strictly greater than $w(\M)$, contradicting the maximality of $\M$.
\end{proof}
So, under the condition of \Lem{cons-mon}, $s_{i+1}$ is $\M$-matched. We can also specify
the comparison relation of $s_{i+1}$, $\M(s_{i+1})$ (as condition \eqref{eq:**}) using an almost identical argument.
Abusing notation, we will denote $\M(s_{i+1})$ as $s_{i+2}$. (This is no abuse if $(s_{i+1}, \M(s_{i+1}))$
is a straight pair.)
\begin{lemma}\label{lem:cons-mon2}
Consider some even index $i$ such that $s_i$ exists.
Suppose conditions \eqref{eq:*} and \eqref{eq:**} held for all even indices $\leq i$.
Then, condition \eqref{eq:**} holds for $i+2$.
\end{lemma}
Before we prove this lemma, we need the following distinctness claim.
\begin{claim} \label{clm:rep} Consider some odd $j$ such that $s_j$ and $\M(s_{j})$ exist.
Suppose condition \eqref{eq:*} and \eqref{eq:**} held for all even $i < j$.
Then the sequence $s_{-1}, s_0, s_1, \ldots, s_j, \M(s_{j})$ are distinct.
\end{claim}
\begin{proof} (If $(s_j, \M(s_{j})) \in \st{\M}{H}$, this is obviously true. The 
challenge is when $\bS_x$ terminates at $s_j$.) The sequence from $s_0$ to $s_{j}$ is an alternating path, so
all terms are distinct. If $s_j \neq y$, then the claim holds. 
Suppose $s_j = y$. Note that $j > 1$, since $(x,y) \notin H$. Since $y \in U(H)$, by \Prop{sub}, $j \equiv 1 \md 4$. Condition \eqref{eq:**}
holds for $j-1$, so $s_{j-1} \prec s_{j} = y$ and by \Cor{sub}, $s_{j-1} \prec s_{j-2}$. Note that $(s_{j-1}, s_{j}) \in H$
and $(x,s_{j})$ is a cross pair. By \Prop{cross}, $x \prec s_{j-1}$ and thus $x \prec s_{j-2}$.
We replace pairs $A = \{(x,y), (s_{j-2},s_{j-1})\} \in \M$ with $(x,s_{j-2})$, and argue that the weight has increased.
We have $w(A) = [f(x) - f(y)] + [f(s_{j-1}) - f(s_{j-2})] = [f(x) - f(s_{j-2})] - [f(y) - f(s_{j-1})]$.
By condition \eqref{eq:*} on $i$, $f(y) = f(s_{j}) > f(s_{j-1})$,  contradicting the maximality of $\M$.
\end{proof}
\begin{proof} (of \Lem{cons-mon2}) By \Lem{cons-mon}, $\M(s_{i+1})$ exists.
Assume $i \equiv 0 \md 4$ (the other case is analogous and omitted). The proof is again by contradiction,
so we assume condition \eqref{eq:**} does not hold for $i+2$. This means $s_{i+2} = \M(s_{i+1}) \prec s_{i+1}$.
Consider sets $E_-(i+2)$ and $E' = E_+(i-2) \cup (s_{i-2},s_{i+2})$. By \Clm{rep}, $s_{-1}, s_0, s_1, \ldots, s_{i+2}$ are distinct.
So $\M' = \M - E_-(i) + E'$ is a valid matching and we argue that $w(\M') > w(\M)$.
By condition \eqref{eq:**} for even $i' < i+2$ and the assumption $s_{i+2} \prec s_{i+1}$.
\begin{eqnarray}
w(E_-(i+2)) & = & [f(s_{0}) - f(s_{-1})] + [f(s_1) - f(s_2)] + [f(s_4) - f(s_3)] + \cdots \nonumber \\
&& \cdots + [f(s_{i-3}) - f(s_{i-2})] + [f(s_{i}) - f(s_{i-1})] + [f(s_{i+2}) - f(s_{i+1})] \nonumber
\end{eqnarray}
Observe how the last term in the summation differs from the trend.
All comparisons in $E_+(i-2)$ are determined by \Prop{ek}, just as we argued
in the proof of \Lem{cons-mon}. The expression for $w(E_+(i-2))$ is basically
given in \eqref{eq:we+}. It remains to deal with $(s_{i-2},s_{i+2})$.
By condition \eqref{eq:**} for $i$, $s_i \prec s_{i-1}$.
Thus, by \Prop{ek}, $s_{i+1} \prec s_{i-2}$.
Combining with the assumption of $s_{i+2} \prec s_{i+1}$, we deduce $s_{i+2} \prec s_{i-2}$.
\begin{eqnarray}
w(E_+(i+2)) & = & [f(s_{1}) - f(s_{-1})] + [f(s_0) - f(s_3)] + [f(s_5) - f(s_2)] + \cdots \nonumber \\
& & \cdots +  [f(s_{i-3}) - f(s_{i-6})] + [f(s_{i-4}) - f(s_{i-1})] + [f(s_{i+2}) - f(s_{i-2})] \nonumber
\end{eqnarray}
The coefficients are identical, except that $f(s_i)$ and $f(s_{i+1})$ do not appear in $w(E_+(i+2))$.
We get $w(\M) - W_- + W_+ = w(\M) + f(s_{i+1}) - f(s_{i})$. By \eqref{eq:*} for $i$, we contradict the maximality of $\M$.
\end{proof}
A direct combination of the above statements yields the main structure lemma.
\begin{lemma} \label{lem:last} Suppose $\bS_x$ contains no violated $H$-pair. Let the last
term by $s_j$ ($j$ is odd). For every even $i \leq j+1$, condition \eqref{eq:**} holds, 
and $s_j$ belongs to a pair in $\sk{\M}{H}$.
\end{lemma}
\begin{proof} We prove the first statement by contradiction. Consider the smallest even $i \leq j+1$
where condition \eqref{eq:**} does not hold. Note that for $i=0$, the condition does hold,
so $i \geq 2$. We can apply \Lem{cons-mon2} for $i-2$, since all even indices at most $i-2$
satisfy \eqref{eq:*} and \eqref{eq:**}. But condition \eqref{eq:**} holds for $i$, completing the proof.
Now apply \Lem{cons-mon} and \Lem{cons-mon2} for $j-1$. Conditions \eqref{eq:*} and \eqref{eq:**} hold for all relevant even indices.
Hence, $s_{j}$ must be $\M$-matched and condition \eqref{eq:**} holds for $j+1$. Since $\bS_x$ terminates
at $s_j$, $s_j$ cannot be $\st{\M}{H}$-matched. Suppose $s_j$ was $\cross{\M}{H}$ matched. Let $j \equiv 1 \md 4$. By \Prop{sub}, $s_j \in U(H)$,
so $s_{j+1} = \M(s_j) \prec s_j$, violating condition \eqref{eq:**}. A similar argument holds when $j \equiv 3 \md 4$.
Hence, $s_j$ must be $\sk{\M}{H}$-matched.
\end{proof}
\section{Monotonicity on  Boolean Hypercube}\label{sec:mono-hc}
We prove \Thm{mono-hc}. Since $\M$ is also is a maximal family of disjoint violating pairs, and therefore, $|\M| \geq \frac{1}{2}\eps_f\cdot 2^n$. 
We denote the set of all edges of the hypercube as $\H$. We partition $\H$ into $H_1,\ldots,H_n$ where $H_r$ is the collection of hypercube edges which differ in the $r$th coordinate. Each $H_r$ is a perfect matching and is adequate. Note that $\st{\M}{H_r}$ is the set of $\M$-pairs which do not differ in the $r$th coordinate.
The $H$-distance is trivially $1$, so $\cross{\M}{H_r}$ is the set of $\M$-pairs that differ in the $r$th coordinate. 
Importantly, $\sk{\M}{H_r} = \emptyset$.
\begin{lemma}\label{lem:mono-hc}
For all $1\leq r\leq n$, the number of violating $H_r$-edges is at least $\cross{\M}{H_r}/2$.
\end{lemma}
\begin{proof} Feed in $\M$ and $H_r$ to the alternating path machinery. Set $X$ to be the set
of all lower endpoints of $\cross{\M}{H_r} \setminus H_r$, so $|X| = |\cross{\M}{H_r} \setminus H_r|/2$. Since $\sk{\M}{H_r} = \emptyset$,
by \Lem{last}, all sequences $\bS_x$ must contain a violated $H_r$-edge. The total number
of violated $H_r$-edges is at least $|X| + |\cross{\M}{H_r} \cap H_r|$.
\end{proof}
The above lemma proves  \Thm{mono-hc}. 
Observe that every pair in $\M$ belongs to some set $\cross{\M}{H_r}$.
The edge tester only requires $O(n/\eps)$ queries, since
the success probability of a single test is at least $$\frac{1}{|\H|}\sum_{r=1}^n \cross{\M}{H_r}/2 \geq |\M|/(n2^{n-2})  \geq \eps/2n.$$ 
\section{Setting up for Hypergrids} \label{sec:setup}
\def\msd{{\tt msd}}
We setup the framework for hypergrid domains. The arguments here are property independent.
Consider domain $[k]^n$ and set  $\ell = \lceil \lg k \rceil$.
We define $\H$ to be pairs that differ in exactly one coordinate, and furthermore, the difference is a power of $2$. 
The tester chooses a pair in $\H$ uniformly at random, and checks the property on this pair.
We partition $\H$ into $n(\ell + 1)$ sets $H_{a,b}$, $1\leq a\leq n$, $0\leq b\leq \ell$.
$H_{a,b}$ consists of pairs $(x,y)$ which differ only in the $a$th coordinate, and furthermore $|y[a] - x[a]| = 2^b$. 
Unfortunately, $H_{a,b}$ is not a matching, since each point can participate in potentially two
pairs in $H_{a,b}$.
To remedy this, we further partition $H_{a,b}$ into $H^0_{a,b}$ and $H^1_{a,b}$.
For any pair $(x,y)\in H_{a,b}$, exactly one among $x[a]\md{2^{b+1}}$\footnote{We abuse notation and define $p \md{2^{b+1}}$ to be $2^{b+1}$ (instead of $0$) if $2^{b+1}\mid p$.} and $y[a]\md{2^{b+1}}$ is $>2^b$ and one is $\leq 2^b$. 
We put $(x,y)\in H_{a,b}$ with $x\prec y$ in $H^0_{a,b}$ if $y[a]\md{2^{b+1}} > 2^b$, and in the set $H^1_{a,b}$ if $1\leq y[a]\md{2^{b+1}} \leq 2^b$.
For example, $H_{1,0}$ has all pairs that only differ by $2^0 = 1$ in the first coordinate.
We partition these pairs depending on whether the higher endpoint has even or odd first coordinate.
Note that each $H^0_{a,b}$ and $H^1_{a,b}$ are matchings. 
We have $L(H^0_{a,b}) = \{x | x[a]\md{2^{b+1}} \leq 2^b\}$ and $U(H^0_{a,b}) = \{y | y[a]\md{2^{b+1}} > 2^b\}$.
The sets are exactly switched for $H^1_{a,b}$.
Because of the matchings are not perfect, we are forced to introduce the notion of adequacy of matchings.
A matching $H$ is adequate if for every violation $(x,y)$, both $x$ and $y$ participate in the matching $H$ (\Def{adeq}).
We will eventually prove the following theorem. 
\begin{theorem} \label{thm:adeq} Let $k$ be a power of $2$. Suppose for every
violation $(x,y)$ and every coordinate $a$, $|y[a] - x[a]| \leq 2^c$ (for some $c$).
Furthermore, suppose that for $b \leq c$, all matchings $H^0_{a,b}, H^1_{a,b}$ are adequate. Then there exists a maximal matching $\M$ of the violation
graph such that the number of violating pairs in $\H$ is at least $|\M|/2$.
\end{theorem}
We reduce to this special case using a simple padding argument. 
The following theorem implies \Thm{mono-hg}.
\begin{theorem} \label{thm:gen-k} Consider any function $f:[k]^n \mapsto R$. 
At least an $\eps_f/(4n(\lceil \log k \rceil + 1)$-fraction of pairs in $\H$ are violations.
\end{theorem}
\begin{proof} Let $\hat{k} = 2^\ell$ be the smallest power of $2$ larger than $4k$. 
Let us construct a function $\hat{f}:[\hat{k}]^n \mapsto \R \cup \{-\infty, +\infty\}$.
Let $\bll$ denote the $n$-dimensional vector all $1$s vector.
For $x$ such that all $x_i \in [\hat{k}/4 + 1, \hat{k}/4 + k-1]$, we set $\hat{f}(x) = f(x - \frac{\hat{k}\cdot\bll}{4})$.
(We will refer to this region as the ``original domain".)
If any coordinate of $x$ is less than $\hat{k}/4$, we set $\hat{f}(x) = -\infty$.
Otherwise, we set $f(x) = +\infty$.
All violations are contained in the original domain.
For any violation $(x,y)$ and coordinate $a$, $|y[a] - x[a]| \leq k < 2^{\ell-2}$.
Let $\hat{\H}$ be the corresponding set of pairs in domain $[\hat{k}]^n$.
For $b \leq \ell-2$ (and every $a$), every point in the original domain participates in 
all matchings in $\hat{\H}$.
So, each of these matchings is adequate. Since every maximal matching
of the violation graph has size at least $\eps_f k^n/2$, by \Thm{adeq},
the number of violating pairs in $\hat{\H}$ is at least $\eps_f k^n/2$.
The matching $\H$ is exactly the set of pairs of $\hat{\H}$
completely contained in the original domain. All violating pairs in $\hat{\H}$
are contained in $\H$. The total size of $\H$ is at most $nk^n(\lceil \log k \rceil + 1)$.
The proof is completed by dividing $\eps_f k^n/4$ by the size of $\H$.
\end{proof}
Henceforth, we will assume that $k = 2^\ell$ and that 
all matchings $H^0_{a,b}, H^1_{a,b}$ are adequate (for $b \leq c$, where
$2^c$ is an upper bound on the coordinate difference for any violation).
\subsection{The potential $\Phi$} \label{sec:phi}
Define $\msd(a)$ of a non-negative integer $a$ to be the largest power of $2$ which divides $a$. That is, $\msd(a) = p$ implies $2^p\divides a$ but $2^{p+1}\notdivides a$. We define $\msd(0) := \ell+1$. For any $x \in \ZZ^n$, define $\Phi(x) = \sum_{c = 1}^n \msd(|x[c]|)$.
Now given a matching $\M$, define the following potential.
\begin{equation}\label{eq:phi}
\Phi(\M) := \sum_{(x,y)\in \M} \Phi(x-y) = \sum_{(x,y) \in \M} \sum_{c=1}^n \msd(|y[c] - x[c]|).
\end{equation}
\noindent
We will choose maximum weighted matchings that also maximize $\Phi(\M)$.
To give some intuition for the potential, note that it is aligned towards picking pairs which differ in as few coordinates as possible (since $\msd(0)$ is large). Furthermore, divisibility by powers of $2$ is favored. 
\medskip
\section{Monotonicity on Hypergrids}\label{sec:mono-hg}
In this section, we prove \Thm{mono-hg}. 
As in the hypercube case, the weight of a pair $(x,y)$ is defined to be $f(x)-f(y)$ if $x\prec y$, and $-\infty$ otherwise. 
We set $\M$ to be a maximum weighted matching that maximizes $\Phi(\M)$.
So $|\M| \geq \eps_f k^n/2$.
Fix $H^r_{a,b}$. 
It is instructive to explicitly see the pairs in $\st{\M}{H^r_{a,b}}$ and $\cross{\M}{H^r_{a,b}}$.
Consider a pair $(x,y)$, $x \prec y$ in these sets. 
\begin{asparaitem}
	\item $\st{\M}{H^r_{a,b}}$: $x[a], y[a] {\md{2^{b+1}}} \leq 2^b$, or $x[a], y[a] {\md{2^{b+1}}} > 2^b$. 
	\item $\cross{\M}{H^r_{a,b}}$: $\msd(|y[a] - x[a]|) = b$, $x \in L(H^r_{a,b})$ (thus $y \in U(H^r_{a,b})$).
\end{asparaitem}
Now we do have skew pairs, and the potential $\Phi$ was designed specifically to handle such pairs.
Note that every pair in $\M$ belongs to some $\cross{\M}{H^r_{a,b}}$.
There exists some $a,b$ such that $\msd(|y[a] - x[a]|) = b$. If $x[a] \md{2^{b+1}} \leq 2^b$,
then $(x,y) \in \cross{\M}{H^0_{a,b}}$, otherwise $(x,y) \in \cross{\M}{H^1_{a,b}}$. 
Therefore, the following lemma directly implies \Thm{adeq}.
\begin{lemma}\label{lem:mono-hg}
For all $r,a,b$, the number of violated $H^r_{a,b}$-pairs is at least $|\cross{\M}{H^r_{a,b}}|/2$.
\end{lemma}
\noindent
\begin{proof} We assume that $H^r_{a,b}$ is adequate.
Feed in $H^r_{a,b}$ and $\M$ to the alternating paths machinery,
with $X$ as the set of lower endpoints in $\cross{\M}{H^r_{a,b}} \setminus H^r_{a,b}$. 
By \Lem{last}, if a sequence $\bS_x$ does not contain a violating $H^r_{a,b}$-pair,
then the last term $s_j$ must belong to $\sk{\M}{H^r_{a,b}}$.
By \Lem{skew}, $\msd(|s_j[a] - \M(s_j)[a]|) > b$. 
But then both $s_j$ and $\M(s_j)$ belong
to $L(H^r_{a,b})$ or $U(H^r_{a,b})$, implying $(s_j, \M(s_j)) \in \st{\M}{H}$. Contradiction.
Every sequence $\bS_x$ contains a violating $H^r_{a,b}$-pair, and the calculation
in \Lem{mono-hc} completes the proof.
\end{proof}
The main technical work is in the proof of \Lem{skew}.
Fix $a,b,r$. For convenience, we lose all superscripts and subscripts. 
\begin{lemma} \label{lem:skew} Suppose $\bS_x$ contains no violated $H$-pair. Let the last term be $s_j$ ($j$ is odd).
Then $\msd(|s_j[a] - \M(s_j)[a]|) > b$.
\end{lemma}
\begin{proof} For convenience, we denote $s_{j+1} = \M(s_j)$.
We prove by contradiction, so $\msd(|s_{j}[a] - s_{j+1}[a]|) \leq b$.
By \Lem{last}, for all even $i \leq j+1$, condition \eqref{eq:**} holds and
$s_j$ belongs to an $H$-skew pair. We will rewire $\M$ to $\M'$ such that weight remains the same
but the potential increases. We will remove the set $E_-(j+1)$ from $\M$ and add the set 
$\hat{E} = E_+(j-1) \cup (s_{j-1},s_{j+1})$. 
Observe that both $E_-(j+1)$ and $\hat{E}$ involve all terms in $s_{-1},\ldots,s_{j+1}$.
We will assume that $j \equiv 1 \md 4$ (the other case is analogous and omitted).
By \eqref{eq:**}, 
$$ w(E_-(j+1)) = [f(s_{0}) - f(s_{-1})] + [f(s_1) - f(s_2)] + [f(s_4) - f(s_3)] + \cdots + [f(s_{j-1}) - f(s_{j-2})] + [f(s_{j}) - f(s_{j+1})]$$
Now for $w(\hat{E})$, all pairs other than $(s_{j-1},s_{j+1})$ have their order decided by \Prop{ek}.
By \eqref{eq:**} for $j-1$ and \Cor{sub} for $j+1$, $s_{j-1} \prec s_j \prec s_{j+1}$.
$$ w(\hat{E}) = [f(s_{1}) - f(s_{-1})] + [f(s_0) - f(s_3)] + [f(s_5) - f(s_2)] + \cdots + [f(s_{j}) - f(s_{j-3})] + [f(s_{j-1}) - f(s_{j+1})]$$
We get $w(E_-(j+1)) = w(\hat{E})$, so the weight stays the same.
It remains the argue that the potential has increased, as argued in \Clm{pot}
\end{proof}
\begin{claim} \label{clm:pot} Suppose $\msd(|s_{j}[a] - s_{j+1}[a]|) \leq b$. Then $\Phi(\hat{E}) > \Phi(E_-(j+1))$.
\end{claim}
\begin{proof} Consider $(s_{j'},s_{j'+1})$ for odd $-1 < j' < j$. 
Both these terms are either in $L(H)$ or $U(H)$. Hence, $\Phi(s_{j'} - s_{j'+1}) = \Phi(H(s_{j'}) - H(s_{j'+1}))
= \Phi(s_{j'-1} - s_{j'+2})$.
So most quantities in $\Phi(E_-(j+1))$ and $\Phi(\hat{E})$ are
identical. 
$$ \Phi(\hat{E}) - \Phi(E_-(j+1)) = \Phi(s_{-1} - s_1) + \Phi(s_{j+1} - s_{j-1}) - [\Phi(s_{-1} - s_0) + \Phi(s_{j} - s_{j+1})] $$
Since $s_1 = H(s_0)$, the points $s_{-1} - s_1$ and $s_{-1} - s_0$ only differ in the $a$th coordinate.
A similar argument works for the remaining terms.
Using $|\cdot|_a$ to denote the absolute value of the $a$th coordinate,
$$ \Phi(\hat{E}) - \Phi(E_-(j+1)) = \msd(|s_{-1} - s_1|_a) + \msd(|s_{j+1} - s_{j-1}|_a) - [\msd(|s_{-1} - s_0|_a) + \msd(|s_{j} - s_{j+1})|_a] $$
Note that $\msd(|s_{-1} - s_0|_a) = b$, by definition, since it lies in $\cross{\M}{H^0_{a,b}}$.
Furthermore $|s_{-1} - s_1|_a = |s_{-1} - H(s_0)|_a = |s_{-1} - s_0|_a - 2^b$, so $\msd(|s_{-1} - s_{1}|_a) > b$.
(Note the strict inequality.)
It suffices to show that $\msd(|s_{j+1} - s_{j-1}|_a) \geq \msd(|s_{j} - s_{j+1}|_a)$. 
Because $s_{j-1} = H(s_j)$, $|s_{j+1} - s_{j-1}|_a$ is either $|2^b + |s_{j} - s_{j+1}|_a|$ or $|2^b - |s_{j} - s_{j+1}|_a|$.
In either case, the assumption $\msd(|s_{j} - s_{j+1}|_a) \leq b$ implies $\msd(|s_{j+1} - s_{j-1}|_a) \geq \msd(|s_{j} - s_{j+1}|_a)$. 
\end{proof}
\section{A pseudo-distance for $(\alpha,\beta)$-Lipschitz}\label{sec:dist}
A key concept that unifies Lipschitz and monotonicity is a
pseudo-distance defined on $\D$.  
The challenge faced in the final proof is tweezing out all
the places in the previous argument where the distance function is ``hidden".
We define a weighted directed graph $\G = (\D,E)$ where $\D$ is the hypergrid $[k]^n$. 
$E$ contains directed edges of the form $(x,y)$, where $\|x-y\|_1 = 1$.
The length of edge $(x,y)$ is gives as follows. If $x \prec y$, the length is $-\alpha$.
If $x \succ y$, the length is $\beta$.
\begin{definition} \label{def:d} The function $\d(x,y)$ between $x,y \in \D$ 
is the shortest path length from $x$ to $y$ in $\G$.
\end{definition}
\noindent
This function is asymmetric, meaning that $\d(x,y)$ and $\d(y,x)$ are
possibly different. Furthermore, $\d(x,y)$ can be negative, so this is
not a distance in the usual parlance of metrics.
Nonetheless, $\d(x,y)$ has many useful properties, which can be proven by 
expressing it in a more convenient form.
Given any $x,y\in \D$, we define $\hcd(x,y)$ to be the $z\in \D$ maximizing $||z||_1$ such that $x\succ z$ and $y\succ z$. 
Note that if $x\succ y$ then $\hcd(x,y) = y$.
\begin{claim} \label{clm:d}
For any $x,y\in \D$, $\d(x,y) = \beta \hd{x}{\hcd(x,y)} - \alpha\hd{y}{\hcd(x,y)}$.
\end{claim}
\begin{proof} Let us partition the coordinate set $[n] = A \cupdot B \cupdot C$
with the following property. For all $i \in A$, $x_i > y_i$.
For all $i \in B$, $x_i < y_i$, and for all $i \in C$, $x_i = y_i$.
Any path in $\G$ can be thought of as sequence of coordinate increments
and decrements.
Any path from $x$ to $y$ must finally decrement all coordinates in $A$,
increment all coordinates in $B$, and preserve coordinates in $C$.
Furthermore, increments add $-\alpha$ to the path length,
and decrements add $\beta$. 
Fix a path, and let $I_i$ and $D_i$ denote the number of increments and decrements
in dimension $i$. For $i \in A$, $D_i = I_i + |x_i-y_i|$, for $i \in B$, $I_i = D_i + |x_i-y_i|$,
and for $i \in C$, $I_i = D_i$. The path length is given by
\begin{eqnarray*} & & \sum_{i \in A} (\beta D_i - \alpha I_i) + \sum_{i \in B} (\beta D_i - \alpha I_i)
+ \sum_{i \in C} (\beta D_i - \alpha I_i) \\
& = & \sum_{i \in A} [\beta|x_i-y_i| + I_i(\beta-\alpha)]
+ \sum_{i \in B} [-\alpha|x_i-y_i| + D_i(\beta-\alpha)] + \sum_{i \in C} I_i(\beta-\alpha)\\
& \geq & \beta\sum_{i\in A} (x_i - y_i) - \alpha\sum_{i\in B} (y_i - x_i)
\end{eqnarray*}
For the inequality, we use $\beta \geq \alpha$. 
Let $z = \hcd(x,y)$. Note that $z_i = \min(x_i,y_i)$.
Consider the path from $x$ that only decrements to reach $z$, and then only increments to reach $y$.
The length of this path is exactly $\beta\sum_{i\in A}(x_i-y_i) - \alpha\sum_{i\in B}(y_i-x_i)$.
\end{proof}
\noindent
It is instructive see the distance for monotonicity and Lipschitz. In the case of monotonicity (when $\alpha=0, \beta=\infty$), 
$\d(x,y) = 0$ if $x\prec y$ and $\d(x,y) = \infty$ otherwise. In the case of Lipschitz, $\d(x,y) = ||x-y||_1$.
\smallskip
\noindent
The next two claims establish some properties of the pseudo-distance. 
\begin{claim}\label{clm:d-prop}
\begin{asparaitem} 
	\item (Triangle equality) Fix $x,y$. Suppose $z$ has the property that for all coordinates $a$, $z[a]$ lies in $[x[a],y[a]]$ or $[y[a],x[a]]$ (whichever
	is valid). Then, $\d(x,y) = \d(x,z) + \d(z,y)$.
	\item (Triangle inequality) $\d(x,y) \leq \d(x,z) + \d(z,y)$.
	\item (Projection)Let $v$ be a vector with a single non-zero coordinate. 
Let $x' = x + v$ and $y' = y + v$. Then $\d(x,y) = \d(x',y')$.
	\item (Positivity) Consider a ``cycle" of distinct points $x_1, x_2, \ldots, x_s, x_{s+1} = x_1$ 
	Then $\sum_{c=1}^s \d(x_c, x_{c+1}) > 0$.
\end{asparaitem}
\end{claim}
\begin{proof}
The triangle equality property follows from \Clm{d}. Suppose $x \succ z \succ y$. 
We have $\hcd(x,y) = y$, $\hcd(x,z) = z$, and $\hcd(y,z) = y$.
Hence, $\d(x,y) = \beta\hd{x}{y}$ $= \beta(\hd{x}{z} + \hd{z}{y})$
$= \d(x,z) + \d(z,y)$. The other case is analogous.
The triangle inequality follows because $\d(x,y)$ is a shortest path length.
For the projection property, let $z = \hcd(x,y)$  and let $z' = \hcd(x',y')$. Note that $z$ and $z'$ also differ only in (say) the $a$th coordinate by the same amount $v_a$. Thus, $\hd{x}{z} = \hd{x'}{z'}$ and $\hd{y}{z} = \hd{y'}{z'}$, implying $\d(x,y) = \d(x',y')$.
For positivity, note that $\d(s_c,s_{c+1})$ is the length of a path in $\G$.
So $\sum_{c=1}^s \d(x_c, x_{c+1})$ is length of a non-trivial cycle in $\G$.
Each coordinate increment adds $-\alpha$ to the length, and a decrement adds $\beta$.
The number of increments and decrements are the same, so the length is a strictly positive multiple of $\beta - \alpha$, a strictly positive
quantity.
\end{proof}
\noindent
The following lemma connects the distance to the $(\alpha,\beta)$-Lipschitz property.
\begin{lemma}\label{lem:l} 
A function is $(\alpha,\beta)$-Lipschitz iff for all $x,y\in \B$, $f(x) - f(y) - \d(x,y) \leq 0$.
\end{lemma}
\begin{proof} Suppose the function satisfied the inequality for all $x,y$.
If $x$ and $y$ differ in one-coordinate by $1$ with $x\prec y$, we get $f(y) - f(x) \leq \beta = \d(y,x)$ and $f(y) - f(x) \geq \alpha = -\d(x,y)$ implying $f$ is $(\alpha,\beta)$-Lipschitz. Conversely, suppose $f$ is $(\alpha,\beta)$-Lipschitz. 
Setting $z=\hcd(x,y)$,
$f(x) - f(z)  \leq \beta \hd{x}{z}$ and $\alpha\hd{y}{z} \leq f(y) - f(z)$. Summing these,
$f(x) - f(y) \leq \beta\hd{x}{z} - \alpha\hd{y}{z} = \d(x,y)$. 
\end{proof}
We give a simple, but important fact about distances related to the function values.
\begin{claim} \label{clm:min} $\min(f(x) - f(y) - \d(x,y), f(y) - f(x) - \d(y,x)) < 0$.
\end{claim}
\begin{proof} Suppose not. Then $f(x) - f(y) - \d(x,y) + $ $f(y) - f(x) - \d(y,x) \geq 0$,
implying $\d(x,y) + \d(y,x) \leq 0$. This violates the positivity of \Clm{d-prop}.
\end{proof}
\noindent
The next lemma is a generalization of \Thm{vc}, which argued that the size of a minimum vertex cover 
is exactly $\eps_f |\D|$. 
We crucially use the triangle
inequality for $\d(x,y)$.
We define an undirected weighted clique on $\D$.
Given a function $f$, we define the weight $w(x,y)$ (for any $x,y \in \B$) as follows: 
\begin{equation}
w(x,y) ~:=  ~~ \max\Big(f(x) - f(y) - \d(x,y), ~~f(y) - f(x) - \d(y,x)\Big) \label{eq:w}
\end{equation}
\noindent
Note that although the distance $\d$ is asymmetric, the weight is symmetric. 
\Lem{l} shows that a function is $(\alpha,\beta)$-Lipschitz iff all $w(x,y)\leq 0$.
Once again, consider the special cases of monotonicity and Lipschitz. For monotonicity, 
$w(x,y) = f(x) - f(y)$ when $x\prec y$ and $-\infty$ otherwise. For Lipschitz, 
$w(x,y) = |f(x) - f(y)| - \hd{x}{y}$.
\noindent
We define the unweighted {\em violation graph} as
$\VG_f = (\D,E)$ where $E = \{(x,y): w(x,y) > 0\}$. 
The following lemma generalizes \Thm{vc} from \cite{FLNRRS02}.
\begin{lemma}\label{lem:vc} The size of a minimum vertex cover in $\VG_f$
is exactly $\eps_f|\D|$.
\end{lemma}
\begin{proof} Let $U$ be a minimum vertex cover in $\VG_f$.
Since each edge in $\VG_f$ is a violation, the points at which the function is modified must intersect all edges, and therefore should form a vertex cover. 
Thus, $\eps_f |\D| \geq |U|$. We show how to modify the function values at $U$ to get a function $f'$ with no violations. 
We invoke the following claim with $V = \D - U$, and $f'(x) = f(x), \forall x \in V$. 
This gives a function $f'$ such that $\Delta(f,f') = |U|/|\D|$. By \Lem{l}, $f'$
is $(\alpha,\beta)$-Lipschitz, and $|U| \geq \eps_f |\D|$. Hence, $|U| = \eps_f|\D|$.
\begin{claim} \label{clm:partial} Consider partial function $f'$ defined on a subset $V \subseteq \D$,
such that for all $\forall x,y \in V$, $f'(x) - f'(y) \leq \d(x,y)$. 
It is possible to fill in the remaining values such that $\forall x,y \in \B$, $f'(x) - f'(y) \leq \d(x,y)$.
\end{claim}
\begin{proof} We prove by backwards induction on the size of $V$. If $|V| = |\D|$, this is trivially
true. Now for the induction step. It suffices define $f'$ for some $u \notin V$.
We need to define $f'(u)$ so that 
$f'(u) - f'(y)\leq \d(u,y)$ and $f'(x) - f'(u) \leq \d(x,u)$ for all $x,y\in V$. It suffices
to argue that
$$ m := \max_{x\in V} \left(f'(x) - \d(x,u)\right) \ \leq \ \min_{y\in V} \left(f'(y) + \d(u,y)\right) =: M$$
Suppose not, so for some $x, y \in V$, $f'(x) - \d(x,u) > f'(y) + \d(u,y)$.
That implies that $f'(x) - f'(y) > \d(x,u) + \d(u,y) \geq \d(x,y)$ (using triangle inequality).
Contradiction, so $m \leq M$.
\end{proof}
\end{proof}
\noindent
The following is a simple corollary of the previous lemma.
\begin{corollary}\label{cor:vc-maxm}
The size of any maximal matching in $VG_f$ is at least $\frac{1}{2}\eps_f|\D|$.
\end{corollary}
\noindent
By a perturbation argument, we can assume that $w(x,y)$ is never exactly zero.
This justifies the strict inequalities used in the monotonicity proofs.
\begin{claim} \label{clm:pert} For any function $f$, there exists a function $f'$
with the following properties. Both $f$ and $f'$ have the same set of violated
pairs, $\eps_f = \eps_{f'}$, and for all $x,y \in \B$, $w_{f'}(x,y) \neq 0$.
\end{claim}
\begin{proof} We will construct a function $f'$ such that $w_{f'}(x,y)$
has the same sign as $w_f(x,y)$. When $w_f(x,y) = 0$, then $w_{f'}(x,y) < 0$.
Since exactly the same pairs have a strictly positive weight, their violation
graphs are identical. By \Lem{vc}, $\eps_f = \eps_{f'}$.
Construct the following digraph $T$ on $\D$. For every $x,y$ such that $f(x) - f(y) - \d(x,y) = 0$, put
a directed edge from $y$ to $x$. Suppose there is a cycle $x_1, x_2, \ldots, x_s, x_{s+1} = x_1$
in this digraph. Then $\sum_{c=1}^s [f(s_c) - f(s_{c+1}) - \d(s_c, s_{c+1})]$
$= -\sum_{c=1}^s \d(s_c, s_{c+1}) = 0$. This violates the positivity of \Clm{d-prop},
so $T$ is a DAG.
Pick a sink $s$. For any $x$, $f(x) - f(s) - \d(x,s)$ is non-zero. Infinitesimally
decrease $f(s)$ (call the new function $f'$). For all $x$, $w_{f'}(x,s)$ has the same
sign as $w_f(x,s)$ and is strictly negative if $w_f(x,s) = 0$.
By iterating in this manner, we generate the desired function $f'$.	
\end{proof}
\submit{\vspace{-15pt}}
\section{Generalized Lipschitz Testing on Hypergrids}\label{sec:lip}
In this section, we prove \Thm{main}.  With the distance $\d(x,y)$ in place, the basic spirit of the monotonicity proofs can
be carried over. The final proof requires manipulations of the distance function.
We do not explicitly have the ``directed" behavior of monotonicity that allows for many of rewiring
arguments. 
The matching $\H$ is the same as in \Sec{setup}. The generalized Lipschitz tester picks a pair $(x,y)\in \H$ at random.
We choose $\M$ to be the maximum weight matching that also maximizes $\Phi(M)$ (as defined by \eqref{eq:phi}).
We again set up the alternating paths as in \Sec{altpaths}, by fixing some matching $H^r_{a,b}$ and taking
alternating paths with $\st{\M}{H^r_{a,b}}$.
We have a minor change that aids in some case analysis.
By \Clm{min}, either $f(x) - f(y) > \d(x,y)$ or $f(y) - f(x) > \d(y,x)$, but not both.
We will show that it suffices to consider only one of these cases.
To that effect, define the set $X$ as follows.
$$ X = \{ x | (x,y) \in \cross{\M}{H^r_{a,b}} \setminus H^r_{a,b}, x \in L(H^r_{a,b}), f(x) - f(y) > \d(x,y)\} $$
(For monotonicity, the last condition is redundant.) As before, the main lemma is the following.
\begin{lemma} \label{lem:sx-gen} For all $x \in X$, $\bS_x$ contains a 
violated $H^r_{a,b}$-pair.
\end{lemma}
We apply some symmetry arguments to show the next lemma, which proves \Thm{adeq}. For convenience, we drop the sub/superscripts
in $H^r_{a,b}$. (Note that we do not lose the $2$ factor here, as compared to \Lem{mono-hc}.)
\begin{lemma} \label{lem:viol-gen} The number of violations in $H$
is at least $\cross{\M}{H}$.
\end{lemma}
\begin{proof} We can classify the endpoints of $\cross{\M}{H} \setminus H^r_{a,b}$ into
the following sets. Consider a generic $(x,y) \in \cross{\M}{H}$ where $x \in L(H)$.
If $f(x) - f(y) > \d(x,y)$, we put $x$ in $X$ and $y$ in $Y$. Otherwise, $f(y) - f(x) > \d(y,x)$,
and we put $x$ in $X'$ and $y$ in $Y'$.
By \Lem{sx-gen}, for $x \in X$, $\bS_x$ has a violated $H$-pair. Consider $x' \in X'$.
Take the function $\hat{f} = -f$ and the $(-\beta,-\alpha)$-Lipschitz property.
By \Clm{d}, the new distance satisfies $\hat{\d}(u,v) = \d(v,u)$. If $f(u) - f(v) > \d(u,v)$,
then $\hat{f}(v) - \hat{f}(u) > \hat{\d}(v,u)$ (and vice versa). Hence, the violation graphs, the weights, $\M$, and the
alternating paths are identical.
Take $x' \in X'$, so it belongs to some $(x',y') \in \cross{\M}{H}$. 
We have $\hat{f}(x) - \hat{f}(y) > \hat{\d}(x,y)$. Applying \Lem{sx-gen} to $\hat{f}$ for the $(-\beta,-\alpha)$-Lipschitz property,
$\bS_{x'}$ has a violated $H$-pair. All in all, for any $x \in X \cup X'$, $\bS_x$ contains 
a violated $H$-pair.
To deal with $Y \cup Y'$, we will first reverse the entire domain, by switching the direction
of all edges in the hypergrid. (Represent this transformation by $\Psi:[k]^n \to [k]^n$, and note
that $\Psi^{-1} = \Psi$.) By the shortest path definition of $\d$,
the new distance satisfies $\hat{\d}(u,v) = \d(\Psi(v),\Psi(u))$. Hence,
we are looking at the $(-\beta, -\alpha)$-Lipschitz property.
The matching $H$ remains the same, but the identities of $L(H)$ and $U(H)$
have switched. Construct function $\hat{f}(x) = -f(\Psi(x))$. If in the original domain $f(u) - f(v) > \d(u,v)$,
then $\hat{f}(\Psi(v)) - \hat{f}(\Psi(u)) > \hat{\d}(\Psi(v),\Psi(u))$ (and vice versa). Again, the alternating path
structure is identical. Consider in the original domain $(x,y) \in \cross{\M}{H}$ where $x \in L(H)$.
In the new domain, $\Psi(y) \in L(H)$. Hence, we can apply the conclusion of the previous paragraph
for all points in $y \in \Psi(Y \cup Y')$, and deduce that $\bS_y$ contains a violated $H$-pair.
Finally, we conclude that every alternating path with an endpoint of $\cross{\M}{H} \setminus H^r_{a,b}$
contains a violated pair. There are at least $|\cross{\M}{H} \setminus H^r_{a,b}|$ such (disjoint) alternating paths.
\end{proof}
\subsection{Preliminary setup} \label{sec:lip-pre}
All the propositions of \Sec{altpaths} hold, since they were independent of the property at hand.
We start by generalizing the monotonicity-specific setup done in \Sec{struct}. We fix some
matching $H^r_{a,b}$, and drop all super/subscripts for ease of notation. 
\begin{proposition} \label{prop:lip-comp} Consider the pairs in $E_-(i)$ and $E^+(i)$. For all
even $0 \leq j \leq i-2$, $\d(s_j,s_{j+3}) = \d(s_{j+1},s_{j+2})$ and $\d(s_{j+3},s_j) = \d(s_{j+2},s_{j+1})$.
\end{proposition}
\begin{proof} By \Prop{sub}, $s_j$ and $s_{j+3}$ both lie in $L(H)$ or $U(H)$. Hence, $s_{j+1} = H(s_j)$
and $s_{j+2} = H(s_{j+3})$ are both obtained by adding or subtracting $2^b$ from the $a$th coordinate.
By the projection property, $\d(s_j,s_{j+3}) = \d(s_{j+1},s_{j+2})$ and $\d(s_{j+3},s_j) = \d(s_{j+2},s_{j+1})$.
\end{proof}
Our aim is to generalize the conditions \eqref{eq:*} and \eqref{eq:**}. The former condition is obtained
by assuming that $(s_i,s_{i+1})$ is not a violation. For monotonicity, this implies a single inequality,
but here, there are two inequalities. It turns out that because we are in the setting where
$w(x,y) = f(x) - f(y) - \d(x,y) > 0$, only one of these is necessary. For
even $i$, if $(s_i,s_{i+1})$ is not a violation, \Cor{sub} implies 
\begin{align} 
\textrm{If} \ i\equiv 0\md 4, f(s_{i+1}) - f(s_{i}) > \alpha 2^b.\nonumber\\
\textrm{If} \ i\equiv 2\md 4, f(s_{i}) - f(s_{i+1}) > \alpha 2^b. \label{eq:circ} \tag{$\circ$}
\end{align}
Nowe we generalize \eqref{eq:**}. The pair $(s_{i-1},s_i)$ is a violation, but we do not
know whether $w(s_{i-1},s_i)$ is $f(s_{i-1}) - f(s_i) - \d(s_{i-1},s_i)$ or $f(s_{i}) - f(s_{i-1}) - \d(s_{i},s_{i-1})$.
The following is the equivalent
of the ordering condition of \eqref{eq:**}.
\begin{align} 
\textrm{If} \ i\equiv 0\md 4, f(s_{i}) - f(s_{i-1}) > \d(s_i,s_{i-1}).\nonumber\\
\textrm{If} \ i\equiv 2\md 4, f(s_{i-1}) - f(s_{i}) > \d(s_{i-1},s_i). \label{eq:2-circ} \tag{$\circ\circ$}
\end{align}
\subsection{The structure lemmas} \label{sec:lip-str}
This lemma is the direct analogue of \Lem{cons-mon}. The proof is also along similar lines.
\begin{lemma}\label{lem:lip-cons}
Consider some even index $i$ such that $s_i$ exists.
Suppose conditions \eqref{eq:circ} and \eqref{eq:2-circ} held for all even indices $\leq i$.
Then, $s_{i+1}$ is $\M$-matched.
\end{lemma}
\begin{proof} The proof is by contradiction. 
Assume $i\equiv 0\md 4$. (The proof for the case $i\equiv 2\md 4$ is similar and omitted.)
As in the proof of \Lem{cons-mon}, we argue that $w(\M') > w(\M)$,
where $\M' = \M - E_-(i) + E_+(i)$. 
By condition \eqref{eq:**},
\begin{eqnarray}
w(E_-(i)) & = & [f(s_{0}) - f(s_{-1}) - \d(s_0, s_{-1})] + [f(s_1) - f(s_2) - \d(s_1, s_2)] \nonumber \\
& & + [f(s_4) - f(s_3) - \d(s_4,s_3)] + [f(s_5) - f(s_6) - \d(s_5,s_6)] + \cdots \nonumber \\
& & + [f(s_{i-3}) - f(s_{i-2}) - \d(s_{i-3},s_{i-2})] + [f(s_{i}) - f(s_{i-1}) - \d(s_i,s_{i-1})] \label{eq:w-lip}
\end{eqnarray}
For $w(E_+(i))$, it suffices to find a lower bound. Since (for any $u,v \in \D$) $w(u,v)$ is the maximum
of two expressions, we can choose the expression to match $w(E_-(i))$ as much as possible.
For a pair $(s_k,s_{k+3})$ in $E_+(i)$, we bound the weight by
$f(s_k) - f(s_{k+3}) - \d(s_k, s_{k+3})$ if $j \equiv 0 \md 4$ and by $f(s_{k+3}) - f(s_k) - \d(s_{k+3},s_k)$ if $j \equiv 2 \md 4$.
This ensure that the coefficients of $f(\cdot)$ are identical to those in \Eqn{w-lip}. 
\begin{eqnarray}
w(E_+(i)) & \geq & [f(s_{1}) - f(s_{-1}) - \d(s_1, s_{-1})] + [f(s_0) - f(s_3) - \d(s_0, s_3)] \nonumber \\
& & + [f(s_5) - f(s_2) - \d(s_5,s_2)] + [f(s_4) - f(s_7) - \d(s_4,s_7)] + \cdots \nonumber \\
& & + [f(s_{i-4}) - f(s_{i-1}) - \d(s_{i-4},s_{i-1})] + [f(s_{i+1}) - f(s_{i-2}) - \d(s_{i+1},s_{i-2})] \label{eq:w+lip}
\end{eqnarray}
Note that only $w(E_+(i))$ involves $f(s_{i+1})$ and only $w(E_-(i))$ involves $f(s_i)$, but all
other $f(\cdot)$ terms have identical coefficients.
To deal with the difference of the distances, we use \Prop{lip-comp}. All the distance terms
in \Eqn{w+lip} except for the first cancel out with an equivalent term in \Eqn{w-lip}.
\begin{eqnarray*}
w(E_+(i)) - w(E_-(i)) & \geq & f(s_{i+1}) - f(s_i) - \d(s_1,s_{-1}) + \d(s_0, s_{-1})
\end{eqnarray*}
Since $(s_0,s_{-1})$ is a cross pair and $s_1 = H(s_0)$, we can use triangle equality to deduce that
$\d(s_0, s_{-1}) - \d(s_1,s_{-1}) = \d(s_0,s_1) = - \alpha 2^b$. 
Combining, $w(E_+(i)) - w(E_-(i)) \geq f(s_{i+1}) - f(s_i) - \alpha 2^b$. By condition \eqref{eq:circ}
for $i$, the RHS is strictly positive. Contradiction.
\end{proof}
Now for analogue of \Lem{cons-mon2} and \Clm{rep}. We will prove the latter first.
\begin{lemma}\label{lem:lip-cons2}
Consider some even index $i$ such that $s_i$ exists.
Suppose conditions \eqref{eq:circ} and \eqref{eq:2-circ} held for all even indices $\leq i$.
Then, condition \eqref{eq:2-circ} holds for $i+2$.
\end{lemma}
\begin{claim} \label{clm:lip-rep} Let $j$ be the last index of $\bS_x$. Suppose
conditions \eqref{eq:circ} and \eqref{eq:2-circ} hold for all even $i < j$. Then the sequence $s_{-1}, s_0 ,s_1, \ldots, s_j, \M(s_j)$ are distinct.
\end{claim}
\begin{proof} By the arguments in \Clm{rep}, it suffices to get a contradiction assuming $s_j = y$.
Since $y \in U(H)$, by \Prop{sub}, $j \equiv 1 \md 4$. 
Note that $s_{j-1} = H(y)$ and $(x,y)$ is a cross pair. Therefore, we have the triangle
equality $\d(x,y) = \d(x,s_{j-1}) + \d(s_{j-1},y) = \d(x,s_{j-1}) - \alpha 2^b$.
We will replace pairs $A = \{(x,y), (s_{j-1},s_{j-2})\} \in \M$ with $(x,s_{j-2})$, and argue that the weight has increased.
Applying condition \eqref{eq:2-circ} for $j-1$,
\begin{eqnarray*}
	w(A) & = & [f(x) - f(y) - \d(x,y)] + [f(s_{j-1}) - f(s_{j-2}) - \d(s_{j-1},s_{j-2})]\\
	& = & f(x) - f(y) + f(s_{j-1}) - f(s_{j-2}) - \d(x,s_{j-1}) + \alpha 2^b - \d(s_{j-1},s_{j-2}) \\
	& \leq & f(x) - f(y) + f(s_{j-1}) - f(s_{j-2}) -\d(x,s_{j-2}) + \alpha 2^b \ \ \ \ \textrm{(triangle inequality)}\\
	& = & [f(x) - f(s_{j-2}) - \d(x,s_{j-2})] - [f(y) - f(s_{j-1}) - \alpha 2^b] \\
	& \leq & w(x,s_{j-2}) - [f(y) - f(s_{j-1}) - \alpha 2^b] 
\end{eqnarray*}
The second term is strictly positive (by condition \eqref{eq:circ} for $j-1 \equiv 0 \md{4}$), contradicting the maximality of $\M$.
\end{proof}
\begin{proof} (of \Lem{lip-cons2}) Assume $i\equiv 0\md 4$. (The proof for the case $i\equiv 2\md 4$ is similar and omitted.)
By \Lem{lip-cons}, $\M(s_{i+1})$ exists, and is denoted by $s_{i+2}$.
The proof is by contradiction, so assume condition \eqref{eq:2-circ} does not hold for $i+2 \equiv 2 \md 4$.
This means $f(s_{i+1}) - f(s_{i+2}) \leq \d(s_{i+1},s_{i+2})$. Since $(s_{i+1},s_{i+2})$ is a violation,
this implies $w(s_{i+1},s_{i+2}) = f(s_{i+}) - f(s_{i+1}) - \d(s_{i+2},s_{i+1})$.
We set $E' = E_+(i-2) \cup (s_{i-2},s_{i+2})$.
We argue that $w(\M') > w(\M)$,
where $\M' = \M - E_-(i+2) + E'$. By \Prop{ek} and \Clm{lip-rep}, $\M'$ is a valid matching.
By condition \eqref{eq:2-circ} for even $k < i+2$ and the above conclusion on $w(s_{i+1},s_{i+2})$,
we get almost the same expression as \eqref{eq:w-lip}.
\begin{eqnarray}
w(E_-(i+2)) & = & [f(s_{0}) - f(s_{-1}) - \d(s_0, s_{-1})] + [f(s_1) - f(s_2) - \d(s_1, s_2)] \nonumber \\
& & + [f(s_4) - f(s_3) - \d(s_4,s_3)] + [f(s_5) - f(s_6) - \d(s_5,s_6)] + \cdots \nonumber \\
& & + [f(s_{i-3}) - f(s_{i-2}) - \d(s_{i-3},s_{i-2})] + [f(s_{i}) - f(s_{i-1}) - \d(s_i,s_{i-1})] \nonumber \\
& & + [f(s_{i+2}) - f(s_{i+1}) - \d(s_{i+2},s_{i+1})] \label{eq:w-lip2}
\end{eqnarray}
For $w(E')$, we follow the same pattern in \eqref{eq:w+lip}.
\begin{eqnarray}
w(E') & \geq & [f(s_{1}) - f(s_{-1}) - \d(s_1, s_{-1})] + [f(s_0) - f(s_3) - \d(s_0, s_3)] \nonumber \\
& & + [f(s_5) - f(s_2) - \d(s_5,s_2)] + [f(s_4) - f(s_7) - \d(s_4,s_7)] + \cdots \nonumber \\
& & + [f(s_{i-3}) - f(s_{i-6}) - \d(s_{i-3},s_{i-6})] + [f(s_{i-4}) - f(s_{i-1}) - \d(s_{i-4},s_{i-1})] \nonumber \\
& & + [f(s_{i+2}) - f(s_{i-2}) - \d(s_{i+2},s_{i-2})] \label{eq:w+lip2}
\end{eqnarray}
By \Prop{lip-comp}, all distance terms
in \Eqn{w+lip} barring the first and last are identical to an equivalent term in \Eqn{w-lip}.
\begin{eqnarray*}
w(E_+(i+2)) - w(E_-(i+2)) & \geq & f(s_{i+1}) - f(s_i) \\
& & - \d(s_1,s_{-1}) - \d(s_{i+2},s_{i-2}) + \d(s_0, s_{-1}) + \d(s_i, s_{i-1}) + \d(s_{i+2},s_{i+1})
\end{eqnarray*}
As in the proof of \Lem{lip-cons},
$\d(s_0, s_{-1}) - \d(s_1,s_{-1}) = \d(s_0,s_1) = - \alpha 2^b$. Furthermore,
\begin{eqnarray*}
- \d(s_{i+2},s_{i-2}) + \d(s_i, s_{i-1}) + \d(s_{i+2},s_{i+1}) & \geq & \d(s_i,s_{i-1}) - \d(s_{i+1},s_{i-2}) \ \ \ \textrm{(triangle inequality)} \\
& = & 0 \ \ \ \textrm{(\Prop{lip-comp})}
\end{eqnarray*}
Combining, $w(E') - w(E_-(i+2)) \geq f(s_{i+1}) - f(s_i) - \alpha 2^b$. This is strictly positive, by condition \eqref{eq:circ}
for $i$. Contradiction.
\end{proof}
We proceed to the analogue of \Lem{last}. Because of the use of distances and potentials,
we require a much simpler statement.
\begin{lemma} \label{lem:lip-last} Suppose $\bS_x$ contains no violated $H$-pair. Let the last
term by $s_j$ ($j$ is odd). For every even $i \leq j+1$, condition \eqref{eq:2-circ} holds.
Furthermore, $s_j$ is $\M \setminus \st{\M}{H}$-matched.
\end{lemma}
\begin{proof} The first part is identical to that of \Lem{last}. 
Condition \eqref{eq:2-circ} holds for $i=0$, and applications of \Lem{lip-cons2}
complete the proof. By \Lem{lip-cons} $s_j$ is $\M$-matched, but being
the last term cannot be $\st{\M}{H}$-matched.
\end{proof}
\subsection{The existence of a violated edge in $\bS_x$}
We show the existence of a violated $H$-edge in $\bS_x$,
proving \Lem{sx-gen}.
Suppose $\bS_x$ has no violated $H$-pair.
By \Lem{lip-last}, $s_j$ is $\M \setminus \st{\M}{H}$-matched.
By the following lemma (analogue of \Lem{skew}) asserts $\msd(s_j[a] - s_{j+1}[a]) > b$, implying $s_j$ is $\st{\M}{H}$-matched.
\begin{lemma} \label{lem:lip-skew} Suppose $\bS_x$ contains no violated $H$-pair. Let the last
term by $s_j$ ($j$ is odd). Then $\msd(s_j[a] - s_{j+1}[a]) > b$.
\end{lemma}
\begin{proof} The proof is analogous to that of \Lem{skew}.
By \Lem{lip-last}, for all even $i \leq j+1$, condition \eqref{eq:2-circ} holds.
By \Clm{lip-rep}, $s_{-1}, s_0, s_1, \ldots, s_j, \M(s_j) = s_{j+1}$ are all distinct.
We rewire $\M$ to $\M'$ by removing $E_-(j+1)$ from $\M$ and adding the set $\hat{E} = E_+(j-1) \cup (s_{j-1},s_{j+1})$.
We will assume that $j \equiv 1 \md 4$ (the other case is analogous and omitted).
By \eqref{eq:2-circ}, we can exactly express $w(E_-(j+1))$.
\begin{eqnarray*}
w(E_-(j+1)) & = & [f(s_{0}) - f(s_{-1}) - \d(s_0, s_{-1})] + [f(s_1) - f(s_2) - \d(s_1, s_2)] \nonumber \\
& & + [f(s_4) - f(s_3) - \d(s_4,s_3)] + [f(s_5) - f(s_6) - \d(s_5,s_6)] + \cdots \nonumber \\
& & + [f(s_{j-1}) - f(s_{j-2}) - \d(s_{j-1},s_{j-2})] + [f(s_{j}) - f(s_{j+1}) - \d(s_j,s_{j+1})]
\end{eqnarray*}
We get a lower bound for $w(\hat{E})$ that matches the $f$ terms exactly.
\begin{eqnarray*}
w(\hat{E}) & \geq & [f(s_{1}) - f(s_{-1}) - \d(s_1, s_{-1})] + [f(s_0) - f(s_3) - \d(s_0, s_3)] \nonumber \\
& & + [f(s_5) - f(s_2) - \d(s_5,s_2)] + [f(s_4) - f(s_7) - \d(s_4,s_7)] + \cdots \nonumber \\
& & + [f(s_{j}) - f(s_{j-3}) - \d(s_{j},s_{j-3})] + [f(s_{j-1}) - f(s_{j+1}) - \d(s_{j-1},s_{j+1})]
\end{eqnarray*}
By \Prop{lip-comp}, the distance terms $\d(s_c,s_{c+3})$ and $\d(s_{c+3},s_c)$ can be
matched to equivalent terms. In the following, we use the equality $\d(s_0,s_{-1}) - \d(s_1, s_{-1}) = -\alpha 2^b$.
\begin{eqnarray*}
w(\hat{E}) - w(E_-(j+1)) & \geq & - \d(s_1, s_{-1}) - \d(s_{j-1},s_{j+1}) + \d(s_0, s_{-1}) + \d(s_{j},s_{j+1}) \\
& \geq & -\alpha 2^b - \d(s_{j-1},s_j) \ \ \ \textrm{(triangle inequality)}\\
& = & -\alpha 2^b - (-\alpha 2^b) = 0 \ \ \ \textrm{(By \Prop{sub}, $j \equiv 1 \md 4$, so $s_j \in U(H)$.)} 
\end{eqnarray*}
So $\M'$ is also a maximum weight matching. 
Observe that the potential $\Phi$ is independent of the property at hand. \Clm{pot} only uses the basic
structure of the alternating paths and is applicable here. It asserts that
if $\msd(s_j[a] - s_{j+1}[a]) \leq b$, then $\Phi(\M') > \Phi(\M)$, contradicting the choice of $\M$.
\end{proof}
\section{Acknowledgements}
Sandia National Laboratories is a multi-program laboratory managed and operated by Sandia Corporation, a wholly owned subsidiary of Lockheed Martin Corporation, for the U.S. Department of Energy's National Nuclear Security Administration under contract DE-AC04-94AL85000. CS is grateful for the support received from the Early Career LDRD program at Sandia National Laboratories.
\bibliographystyle{acmsmall}
\bibliography{derivative-testing}
\end{document}